\documentclass{config/llncs}

\usepackage[margin=1.3in]{geometry} 

\pdfoutput=1 

\usepackage[T1]{fontenc}
\usepackage[english]{babel}
\usepackage[autostyle=true,german=quotes]{csquotes}

\usepackage{xspace}
\usepackage{xparse}
\usepackage{mleftright}

\usepackage{amssymb}
\usepackage{stmaryrd}
\usepackage{mathtools}

\usepackage{thm-restate}
\usepackage{tcolorbox}

\usepackage{minibox}
\usepackage[
pdftitle={Paper},
pdfsubject={},
pdfauthor={},
pdfkeywords={},
pdfcreator={latexmk -pdf},
colorlinks=true,
]{hyperref}

\usepackage[style=alphabetic,giveninits=true,backend=bibtex,isbn=false,doi=true,dateabbrev=false,url=false,natbib=true,maxcitenames=3,maxbibnames=6]{biblatex}

\usepackage{relsize}
\usepackage{graphicx}

\usepackage{booktabs} 
\usepackage{longtable}
\usepackage{tabularx}

\usepackage{bm}
\usepackage{nicefrac}

\usepackage[title]{appendix}
\usepackage{float}

\usepackage[titlenumbered, algoruled, linesnumbered, vlined]{algorithm2e}

\usepackage[capitalize]{cleveref}
\usepackage[n, advantage, operators, sets, adversary, landau, lambda, probability, notions ,logic, ff, mm, primitives, events, complexity, asymptotics, keys]{cryptocode}

\usepackage[tableposition=top, labelfont=bf, font={small}]{caption}
\usepackage[]{subcaption}
\usepackage[nopostdot,nomain,acronym,symbols,nogroupskip,nonumberlist]{glossaries}

\usepackage{framed}
\usepackage[inline]{enumitem}

\usepackage{placeins}

\usepackage{fnpct}
\usepackage[colorinlistoftodos,prependcaption,textsize=small]{todonotes}

\usepackage{scalerel}
\usepackage{stackengine}
\usepackage{rotating}

\usepackage{calc} 

\usepackage{eqparbox}
\usepackage{adjustbox}

\usepackage{comment}

\newif\ifconf

\conffalse

\newcommand{\onlyeprint}[1]{\ifconf\else#1\fi}

\DeclarePairedDelimiterX\Prc[2]{\Pr[}{]}{#1\mspace{1.5mu}\vert\mspace{1.5mu}#2}         

\newcommand{\N}{\mathbb{N}}

\newcommand{\Exp}{\mathsf{Exp}}

\newcommand{\msg}{\ensuremath{\mathsf{m}}}
\newcommand{\ct}{\ensuremath{\mathsf{ct}}}
\newcommand{\msk}{\ensuremath{\mathsf{msk}}}
\newcommand{\mpk}{\ensuremath{\mathsf{mpk}}}
\newcommand{\ek}{\ensuremath{\mathsf{ek}}}
\newcommand{\dk}{\ensuremath{\mathsf{dk}}}

\newcommand{\Xgate}{\ensuremath{\mathsf{X}}}
\newcommand{\Zgate}{\ensuremath{\mathsf{Z}}}
\newcommand{\UEQ}{\ensuremath{\mathsf{UEQ}}}

\renewcommand{\bin}{\ensuremath{\{0,1\}}}

\newcommand{\IdFE}{\ensuremath{\mathsf{IdFE}}}

\newcommand{\OneFE}{\ensuremath{\mathsf{OneQFE}}}
\newcommand{\TwoFE}{\ensuremath{\mathsf{TwoFE}}}

\newcommand{\Enc}{\ensuremath{\mathsf{Enc}}}
\newcommand{\Dec}{\ensuremath{\mathsf{Dec}}}
\newcommand{\Setup}{\ensuremath{\mathsf{Setup}}}
\newcommand{\Keygen}{\ensuremath{\mathsf{KeyGen}}}
\newcommand{\Sim}{\ensuremath{\mathsf{Sim}}}
\newcommand{\Encode}{\ensuremath{\mathsf{Encode}}}
\newcommand{\Decode}{\ensuremath{\mathsf{Decode}}}
\newcommand{\QRE}{\ensuremath{\mathsf{QRE}}}
\newcommand{\QFE}{\ensuremath{\mathsf{QFE}}}

\newcommand{\FE}{\ensuremath{\mathsf{FE}}}

\newcommand{\MIFE}{\ensuremath{\mathsf{MIFE}}}
\newcommand{\QMIFE}{\ensuremath{\mathsf{QMIFE}}}
\newcommand{\TP}{\ensuremath{\mathsf{TP}}}

\newcommand{\Obf}{\ensuremath{\mathsf{Obf}}}
\newcommand{\Eval}{\ensuremath{\mathsf{Eval}}}

\newcommand{\eps}{\varepsilon}

\newcommand{\T}{\ensuremath{\text{TD}}}  
\newcommand{\Tr}{\ensuremath{\text{Tr}}} 

\newcommand{\Adv}{\ensuremath{\mathcal{A}}}
\newcommand{\Dist}{\ensuremath{\mathcal{D}}}
\newcommand{\Hilb}{\ensuremath{\mathcal{H}}}
\newcommand{\AdvB}{\ensuremath{\mathcal{B}}}

\mathtoolsset{mathic}

\usetikzlibrary{matrix,shapes,arrows,positioning,chains,calc,fit, graphs, arrows.meta, quotes, angles}

\SetAlgorithmName{Protocol}{pro}{List of Protocols}
\DontPrintSemicolon
\SetAlgoCaptionSeparator{.}
\DecMargin{0.5\parindent}
\SetAlCapSkip{0.3ex}
\SetAlgoNlRelativeSize{-1}

\crefname{algorithm}{Game}{Games}
\Crefname{algorithm}{Game}{Games}
\crefname{equation}{}{}
\Crefname{equation}{}{}
\crefname{enumi}{item}{items}

\graphicspath{{../figure/}}

\bibliography{cryptobib/abbrev0,cryptobib/crypto}
\DeclareBibliographyAlias{misc}{article}

\let\origleft\left
\let\origright\right
\renewcommand\left{\mathopen{}\mathclose\bgroup\origleft}
\renewcommand\right{\aftergroup\egroup\origright}

\let\doendproof\endproof
\renewcommand\endproof{~\hfill\qed\doendproof}

\definecolor{lightskyblue}{rgb}{0.53,0.81,0.98}

\newlength\glsnamewidth
\newlength\glssymbolwidth
\makeglossaries
\glsdisablehyper

\makeatletter
 \def\@textbottom{\vskip \z@ \@plus 1pt}
 \let\@texttop\relax
\makeatother
\AtEveryBibitem{%
	\ifentrytype{proceedings}{%
		\clearlist{location}%
		\clearname{editor}%
		\clearfield{pages}%
		\clearfield{series}%
		\clearfield{number}%
		\clearfield{year}
		\clearlist{publisher}
	}{}%
	\ifentrytype{inproceedings}{%
		\clearlist{location}%
		\clearname{editor}%
		\clearfield{pages}%
		\clearfield{series}%
		\clearfield{number}%
		\clearfield{volume}%
		\clearfield{year}
		\clearlist{publisher}
	}{}%
	\ifentrytype{article}{%
		\clearlist{location}%
		\clearname{editor}%
		\clearfield{pages}%
		\clearlist{publisher}
	}{}%
	\clearfield{urlyear}
	\clearfield{month}
}

\setlist[enumerate]{wide=.5\parindent,listparindent=1em,labelindent=\parindent}

\title{Unclonable Functional Encryption}

 \author{}
\institute{}
\onlyeprint{
     \author{%
         Arthur Mehta\inst{3} \and Anne Müller\inst{1,2}
     }
     \institute{
     CISPA Helmholtz Center for Information Security, Saarbrücken, Germany\\
         {\scriptsize{\email{anne.mueller@cispa.de}}} \and
         Graduate School of Computer Science, Saarland University, Germany \and
         Department of Mathematics and Statistics, University of Ottawa, Ottawa, Canada\\ {\scriptsize{\email{amehta2@uottawa.ca}}}
     }
}
\usepackage{enumitem}
\setlist[itemize]{label=\textbullet}
\begin{document} 
\pagestyle{plain}
\onlyeprint{
\let\oldaddcontentsline\addcontentsline
\def\addcontentsline#1#2#3{}
}
\maketitle
\onlyeprint{
\def\addcontentsline#1#2#3{\oldaddcontentsline{#1}{#2}{#3}}
}

\begin{abstract}
In a functional encryption (FE) scheme, a user that holds a ciphertext and a function-key can learn the result of applying the function to the plaintext message. Security requires that the user does not learn anything beyond the function evaluation. We extend this notion to the quantum setting by providing definitions and a construction for a quantum functional encryption (QFE) scheme which allows for the evaluation of polynomialy-sized circuits on arbitrary quantum messages. Our construction is built upon quantum garbled circuits \cite{BY20}.

We also investigate the relationship of QFE to the seemingly unrelated notion of unclonable encryption (UE) and find that any QFE scheme \emph{universally} achieves the property of unclonable functional encryption (UFE). In particular we assume the existence of an unclonable encryption scheme with quantum decryption keys which was recently constructed by \cite{AKY24}. Our UFE guarantees that two parties cannot simultaneously recover the correct function outputs using two independently sampled function secret keys. As an application we give the first construction for public-key UE with variable decryption keys. 


Lastly, we establish a connection between quantum indistinguishability obfuscation (qiO) and quantum functional encryption (QFE); Showing that any multi-input indistinguishability-secure quantum functional encryption scheme unconditionally implies the existence of qiO.
\end{abstract}

\onlyeprint{
\newpage
{
\hypersetup{linkcolor=black}
\setcounter{tocdepth}{2}
\tableofcontents
}
\newpage
}

\section{Introduction}
The development of Functional Encryption (FE) marks a significant evolution in cryptography, enabling a more nuanced and controlled access to encrypted data \cite{oneil10, TCC:BonSahWat11}. Traditional public-key encryption allows either full decryption or none at all, a model insufficient for many modern applications, such as cloud services, where selective access to data is essential. FE addresses this by allowing decryption keys to reveal only specific functions of the encrypted data. 

In more detail, an FE scheme for a family of functions $\mathcal{F}$ enables a specialized form of decryption that takes as input both a ciphertext $\ct$ and a function key $\sk_f$ and outputs the evaluation $f(\msg)$ on the plain text $\msg$. The security of the scheme ensures that an adversary in possession of $(\ct, \sk_f)$ cannot recover additional information beyond  $f(\msg)$.

A broad goal within quantum cryptography aims to generalize various cryptographic tools into the quantum setting. This includes works studying verifiable delegation \cite{RUV13, Gri19}, randomized encodings and garbled circuits \cite{BY20}, and quantum indistinguishablity obfuscation (qiO) \cite{LC:BroKaz21}. Another approach explores new functionalities uniquely achievable in the quantum setting, such as unclonable encryption (UE) \cite{BL20}, where an adversary is unable to create two ciphertexts that both decrypt to the same message as the original ciphertext.

While the works mentioned above demonstrate the tremendous progress made in the field, there remain significant open challenges. Prior to this work, a formal treatment and secure construction of quantum functional encryption (QFE) had not been provided. Instead, \cite{BY20} suggests QFE as a potential application of quantum garbled circuits. Additionally, although there has been some progress, a complete construction for either qiO or UE remains an open problem. We explore how QFE can advance these topics.

\subsubsection{Summary of Results.} Our results on the topics of QFE, UE, and qiO are as follows: 
\begin{enumerate}[align= left, leftmargin=2.8em]
     \item We give the first formal definitions of QFE, covering both simulation and indistinguishability-based security. Our treatment spans adaptive and non-adaptive models, as well as multi-message, multi-query, and multi-input scenarios, addressing all key variants of functional encryption.
    \vspace{0.5em}
    \item We use quantum garbled circuits to give the first construction of single-query, adaptively simulation-secure QFE.
    \vspace{0.5em}
    \item We present a universal construction for unclonable QFE, showing that any QFE scheme can be made unclonably secure, assuming the existence of UE. As a corollary, we use this to obtain the first indistinguishable-uncloneable secure public-key encryption scheme with variable decryption keys. \vspace{0.5em}
    \item Laslty, we establish a connection between quantum indistinguishability obfuscation (qiO) and QFE; Showing that any multi-input indistinguishability-secure quantum functional encryption scheme unconditionally implies the existence of qiO.
\end{enumerate}

\subsection{Quantum Functional Encryption}

In this work, we formally define Functional Encryption in the quantum setting, referred to as Quantum Functional Encryption (QFE). At a high level, a QFE scheme for a class of circuits $\mathcal{C}$ allows for selective decryption with respect to function keys $\sk_C$, which must satisfy two key properties: \emph{correctness} and \emph{security}. The correctness property ensures that decryption returns $C(\rho_\msg)$ for all $C \in \mathcal{C}$ and states $\rho_\msg$,
and the security property guarantees that no additional information is revealed beyond $C(\rho_\msg)$.

While correctness in QFE is relatively straightforward, defining security requires more nuanced attention. Security can be analyzed through two primary frameworks: simulation-based security (SIM-security) and the generally weaker indistinguishability-based security (IND-security). Both approaches have further distinctions between adaptive and non-adaptive versions, whether an adversary has a single or multiple challenge ciphertexts, and depending on whether the adversary obtains one or more function keys. The formal definitions and detailed treatments are presented in \cref{qfe:sec} and \cref{sec:Apendix(Multi-Input)}. Below we provide the basic structure of QFE and outline notions of correctness as well as SIM-security and IND-security.

\subsubsection{QFE.$(\Setup, \Keygen, \Enc, \Dec)$}
$\Setup(1^\lambda)$, takes as input the security parameter $\lambda$, and outputs a master public key $\mpk$, and a master secret key $\msk$. Given $\msk$ and a circuit $C \in \mathcal{C}$, the key generation algorithm, $\Keygen(\msk, C)$, produces a secret function key $\sk_C$. Encryption, $\Enc(\mpk, \rho_\msg)$, uses $\mpk$ and outputs a ciphertext $\rho_\ct$. Finally, the decryption algorithm, $\Dec(\sk_C, \rho_\ct)$, takes a function key $\sk_C$ and the ciphertext $\rho_\ct$, and outputs a quantum state.

\subsubsection{Correctness.} Correctness requires that for all messages $\rho_\msg$, circuits $C \in \mathcal{C}$ and random coins used by $\Enc$ and $\Setup$ it holds that 
    \begin{align*}
        &C(\rho_\msg) = \Dec(\sk_C,\rho_\ct).
    \end{align*}
 
As outlined in \cref{qfe:sec} we additionally require correctness to respect correlation with possible side information.

\subsubsection{Simulation Security.}  Simulation security is formalized by comparing the output of two experiments: in the real experiment, the adversary interacts with the actual encryption scheme to produce an encryption of a chosen message, and choice of function key(s) $sk_C$. In the ideal experiment, a simulator is given access to the function key $\sk_C$ and the image state $C(\rho_\msg)$, and produces a ciphertext without access to the underlying message. The scheme is called simulation secure, abbreviated as SIM-secure, if the outputs of these two experiments are computationally indistinguishable. A QFE scheme is further said to be \emph{adaptively} simulation secure if the adversary can either choose the message first and then the function secret key or the other way around. 

The formal definition of simulation security in the restricted setting, where the adversary holds only a single ciphertext and single function key, is provided in \cref{qfe:def:simsecurity}.

\subsubsection{Indistinguishability Security.} In the classical setting, IND-security is defined with respect to \emph{admissible queries}. Specifically, an adversary holding a function key $\sk_f$ for some function $f$ cannot distinguish between encryptions of two admissible queries, meaning that $f(\msg_0) = f(\msg_1)$. Adapting IND-security to the quantum setting introduces some challenges, particularly in defining admissible queries. A first naive approach would be to require the trace distance of outputs states $C(\rho_{\msg_{1}})$ and $C(\rho_{\msg_{1}})$ to be suitably close in order for them to be admissible. However, as we discuss in \cref{qfe:sec:inddef} this approach is insufficient to prevent attacks based on quantum side information. 

An alternative, approach would be to take into account the internal state of an adversary and thereby restricting quantum side information. Although such an approach may be be useful in some applications, such as when the messages are are not chosen by the adversary, it remains too restrictive for many use cases. Instead, in \cref{qfe:def:admis} we introduce a notion of admissible queries which allows an adversary to be entangled with part of the message. As in the case in the classical setting our notion of sim-security is generally stronger and, we show that it implies IND-security.

\subsubsection{Multi-message Security.} More generally we also consider the notion of SIM-security and IND-security in the context where an adversary has access to numerous ciphertexts. In \cref{sec:Apendix(Multi-Input)} we provide an extension of the SIM-security from \cref{qfe:def:simsecurity} to the multi-message setting. As in the classical case, we show in \cref{qfe:lem:singletomulti} that any non-adaptive single-query simulation-secure scheme with classical secret keys is also multi-message simulation-secure. 

\subsubsection{Multi-query/Collusion Security}
In the classical setting, functional encryption schemes often require security to hold even in the presence of colluding key holders. A malicious user should not be able to combine several function keys to extract unauthorized information. More formally, a group of users holding secret keys $\sk_{C_1}, \dots, \sk_{C_q}$, along with an encryption of $\msg$, should only be able to learn $C_1(\msg), \dots, C_q(\msg)$, and nothing more about $\msg$. This scenario is often referred to as "collusion resistance." In our work, we refer to this property as \emph{multi-query security}. We note, however, that classical simulation-secure FE is not achievable against an adversary who may possesses an unbounded number of function keys, a scenario sometimes referred to as unbounded collusion \cite{AGVW13}\footnote{Assuming the existence of a family of weak pseudo-random functions.}.

In the quantum setting, the no-cloning theorem makes it unclear to what extent collusion is possible and presents challenges to formalising multi-querry security. In particular, without several copies of the underlying ciphertext it may not be possible to obtain several evaluations. In \cref{qfe:sec:qmife}, we introduce a more general form of QFE called \emph{quantum multi-input functional encryption} (QMIFE). This framework extends our treatment of both simulation-based security and indistinguishability-based security, encompassing multi-query security as a special case. Below, in \cref{QMIFE:subsec:intro}, we provide an overview of QMIFE and discuss how IND-security and SIM-security can be adapted to QMIFE.

\subsection{QFE for Poly-sized Circuits}

In the classical setting, it is known that a \emph{non-succinct} form of FE can be constructed using a cryptographic primitive known as randomized encodings (RE). Specifically, \cite{CCS:SahSey10, C:GorVaiWee12} show that any RE scheme which possess the additional property of being \emph{decomposable}, can be used to construct an FE scheme for the class of polynomial-sized circuits. Here the constructed FE scheme is considered non-succinct as the size of the ciphertext must be at least as large as the size of allowable circuits. 

\subsubsection{Randomized Encodings}
A randomized encoding (RE) of a function $f$ is a probabilistic function $\hat{f}$ such that, for any input $x$, the value of $f(x)$ can be recovered from $\hat{f}(x)$, but no additional information about $f$ or $x$ is revealed. An RE scheme is called \emph{decomposable} if a function $f$ and a sequence of inputs $(x_1, \dots, x_n)$ can be encoded in two parts: an \emph{offline} part $\hat{f}_{\text{off}}$, which depends only on $f$ and some randomness $r$, and an \emph{online} part $\hat{f}_i$, which depends on each input $x_i$ and the same randomness $r$. We write DRE for RE schemes which satisfy this proprety. 

In \cite{BY20} Brakerski and Yuen both define and give a construction for decomposable RE in the quantum setting called the \emph{Quantum Garbled circuit} (QGC) scheme. Our first main result presents a constuction for QFE based on QGC. 
\begin{theorem}[Informal]
Given a QGC scheme with perfect correctness and a public key encryption scheme there exists a single-query adaptive SIM-secure QFE scheme for the class of polynomial-sized circuits. 
\end{theorem}  

The formal statement and construction of our QFE scheme is given in \cref{qfe:sec:polyqfe}. Similar to the classical constructions given in \cite{CCS:SahSey10, C:GorVaiWee12}, our scheme is not succinct. While succinct FE is needed for many applications, such as delegated computation, we show that the our QFE scheme can be used to obtain the first public-key unclonable encryption scheme with variable decryption keys. This in turn provides several applications such as private-key quantum money. Details on our applications to unclonable cryptography are discussed in \cref{UQFE:subsec:intro}. Below we outline our construction for QFE and highlight the specific challenges which present in the quantum setting.

\subsubsection{Outline of QFE Construction} 
The basic observation that enables the construction of FE from garbled circuits is their decomposability. It allows one to decouple the circuit and input by viewing both as inputs to a universal circuit. Here a universal circuit $U$ takes as input a circuit description $C$ and state $\rho_{\msg}$ and outputs $C(\rho_{\msg})$. Due to the decomposability property the encoding of the classical circuit description and the encoding of the quantum input can be handled separately. In this way, using a decomposable RE scheme for a universal circuit, combined with a restricted form of functional encryption for pairs of circuits, enables functional encryption for all polynomial-sized circuits.

While the above construction is fairly straightforward to translate into the quantum setting using the QGC, more difficulty arises when creating adaptive security. In the adaptive security setting the adversary can first request a ciphertext and then a secret key for an arbitrary function. Since the simulator is not allowed learn the message which the adversary selected, the simulator only obtains the output of the circuit evaluation during the second stage. Therefore the simulator needs to first create an 'empty' ciphertext and later provide a secret key that opens the ciphertext to the correct value. Techniques for handling the classical part adaptively are well known but they cannot be applied to the quantum part. To resolve this we employ a 'trick' inspired by the concept of computation trough teleportation.

We describe the classical and quantum techniques to achieve this for a single bit or qubit repectively. For a classical message an 'empty' ciphertext can be created by encrypting the bit $0$ and the bit $1$ in two separate slots of the ciphertext and later revealing the key for only one slot.
Clearly in the quantum setting we cannot enumerate all possible values a single qubit can take. Instead the uniquely quantum phenomenon of teleportation can help us achieve such a construction. The simulator encrypts one qubit of an EPR pair $\sigma^{AB} = \frac{1}{\sqrt{2}} (|0\rangle^A |0\rangle^B + |1\rangle^A |1\rangle^B)$ pair using the quantum one time pad: $|\ct\rangle = X^a Z^b \sigma^A$.

Later when the simulator learns the output state $\rho$ it teleports the state into the ciphertext. This results into a randomization of the ciphertext since now the state $X^{a'} Z^{b'} \rho$ is contained resulting in the ciphertext $X^a Z^b X^{a'} Z^{b'} \rho = X^{a \xor a'} Z^{b \xor b'} \rho$ where $(a',b')$ are the teleportation correction keys.  We can then use a classical ciphertext as described above to reveal the keys $(a \xor a', b \xor b')$.

\subsection{Unclonable QFE}\label{UQFE:subsec:intro}
As an application of our $\QFE$ scheme we obtain a novel form of unclonable encryption (UE), which we call \emph{unclonable functional encryption}.  

\subsubsection{Unclonable Encryption} Unclonable encryption, first introduced by Broadbent and Lord \cite{BL20}, is an encryption scheme that leverages the no-cloning theorem to achieve a novel cryptographic functionality. Specifically, it guarantees that an adversary in possession of a ciphertext $\rho_{ct}$ cannot generate two states, $\rho_B$ and $\rho_C$, that both correctly decrypt to the same message $\msg$. This is formalised in the following security game with a tripartite adversary $\Adv = (A,B,C)$. In the first phase $A$ receives a ciphertext $\rho_\ct$ that enrcypts a message $\msg$ and has to produce a state $\rho_{BC}$ by applying an arbitrary quantum channel. In the second phase $B$ and $C$ are activated, they receive the state $\rho_B$ and $\rho_C$ respectively and each get a copy of the decryption key. They win the experiment if both $B$ and $C$ correctly guess the message $\msg$. The strongest security notion for unclonable encryption is unclonable-indistinguishability security which allows $A$ to choose two messages $\msg_0, \msg_1$. To win the game $B$ and $C$ have to both guess correctly which of these message was encrypted. A second weaker notion, often simply referred to as unclonable security, weakens the win condition of the security game to requiring $B$ and $C$ to simultaneously recover a random message in it's entirety. In our work all definitions are with respect to the stronger notion of unclonable-indistinguishable security.

\subsubsection{Vairable-key UE} While currently there is no provably secure construction for the strongest notion of UE there do exists several weaker variants which have allowed for more progress. Kundu and Tan consider one such variant called unclonable encryption with variable keys \cite{KT20}. Their modified version of UE allows a ciphertext to be decrypted using multiple decryption keys, with each adversary in a cloning attack receiving an independently generated key. In the device-independent setting \cite{KT20} give a construction for secret-key unclonable encryption with variable keys. While this scheme enjoys device indepedence it only obtains the weaker notion unclonable security.

Kundi and Tan also further outline that although weaker than UE such a scheme is still useful for known applications of UE such as quantum money. A private-key quantum money scheme can be constructed from unclonable encryption as follows: A banknote is created by creating a ciphertext of a random message. The bank holds a deryption key and can verify the message by decrypting it. In the case of unclonable encryption with variable decryption keys each bank that needs to verify the banknote independently samples a decryption key.

\subsubsection{Unclonable QFE} In this work, we introduce a novel cryptographic primitive called \emph{Unclonable Functional Encryption} (Unclonable QFE), which combines the security properties of QFE with the unclonable security characteristics of UE. The formal definition of Unclonable QFE is provided in  \cref{qfe:def:UEQFE}, where we extend the security requirements of a QFE scheme to include unclonable security. Our approach builds on the familiar security game from UE with some key modifications. In the first phase, the underlying message is encrypted using a QFE scheme. After an adversarial splitting channel is applied, in the second phase, two adversaries, $B$ and $C$, each receive independently generated function secret keys for some circuit. Our new security notion ensures that both $B$ and $C$ cannot simultaneously guess which of the two challenge messages was encrypted, thus preserving unclonability in the functional encryption setting. 
Notably, we allow for the encryption of quantum messages. While UE is usually concerned with the protection of classical messages we maintain the properties of functional encryption for quantum messages while at the same time achieving unclonability for any message, classical or quantum. In \cref{qfe:thm:UEQFE} we prove that such a scheme can be constructed from any QFE scheme, such as our construction from \cref{qfe:sec:polyqfe}, assuming the existence of an unclonable encryption scheme which allows for quantum decryption keys, such as that given by \cite{AKY24}.

\begin{theorem}[Informal]
    Any single-query QFE scheme for n-qubit messages is an unclonable-indistinguishable secure functional encryption scheme with variable decryption keys assuming an unclonable-indistinguishable secure encryption scheme with quantum decryption keys for single bit messages.
\end{theorem}

When the function secret keys are fixed to be the identity circuit this implies a public-key unclonable-indisintguishable secure encryption scheme with variable decryption keys. In contrast to the standard definition of unclonable encryption here the $\Keygen$ algorithm is run twice to produce independently sampled secret keys. We assume that the randomness can be chosen in such a way that the same encryption key is produced with different decryption keys. 

\begin{corollary}[Informal] There exists a public-key unclonable-indistinguishable secure encryption scheme with variable decryption keys for n-bit messages assuming a single-query QFE scheme and an unclonable-indistinguishable secure encryption scheme with quantum decryption keys for single bit messages.
\end{corollary}

\subsubsection{Outline of the Unclonable QFE Construction}
The construction is inspired by \textcite{HKNY23} and the universal construction of Waters and Wichs \cite{C:WatWic23}. In  \cite{HKNY23} they show a universal plaintext expansion result for unclonable encryption: A construction based solely on quantum randomized encodings is a mulit-bit unclonable encryption scheme if there exists a single bit unclonable-indistinguishable secure encryption scheme. Unfortunately, the existence of such a scheme is not yet known in the plain model. 
Therefore we instead rely on an unclonable encryption scheme with quantum decryption keys which was recently constructed by \textcite{AKY24}.

The idea of our construction is that the ciphertext has two modes indicated by a flag bit in the plaintext. In the first mode $(f=0)$, which is the mode the real encryption procedure uses, the plaintext is simply encrypted under the $\QFE$ scheme and padded to a certain length:
$$ \rho_{\ct_0} = \QFE.\Enc(\rho_\msg \otimes |0\rangle\langle0|^{O(\lambda)} \otimes |f=0\rangle\langle f=0|)$$

To prove security we want to reduce to the unclonable encryption scheme with quantum secret keys $\UEQ$. Therefore we show that, due to the security of the $\QFE$ scheme, the first ciphertext is indistinguishable to the following ciphertext which makes use of the $\UEQ$ scheme. Let $\ek, |\dk\rangle$ be the encryption and decryption keys of the $\UEQ$ scheme and let $\rho_{UE} \leftarrow \UEQ.\Enc(1^\lambda, b)$ be an unclonable encryption of a bit $b \leftarrow \bin$. Define $\rho_{\msg_b} = \rho_\msg$ and $\rho_{\msg_{1-b}}$ an arbitrary n-qubit state. Then a ciphertext in the second mode is created as follows:

$$\rho_{\ct_1} = \QFE.\Enc(\rho_{\msg_0} \otimes \rho_{\msg_1} \otimes |\dk\rangle \langle\dk| \otimes \rho_{UE}\otimes  |f=1\rangle\langle f=1|) $$

Now we can define a class of circuits $U(C,\cdot)$ that checks the last bit of the message and in the case of $f=0$ outputs the message $C(\rho_{\msg})$. In case of $f =1$ the circuit decrypts the $\rho_{UE}$ ciphertext to get $b$, selects the message $\msg_b$ and outputs $C(\rho_{\msg_b})$. Indistinguishability of the ciphertexts $\rho_{\ct_0}$ and $\rho_{\ct_1}$ follows from the security of the QFE scheme. 

During the reduction we encounter the issue that we have to create the $\QFE$ ciphertext before we learned the decryption key $|\dk\rangle$ for the single-bit unclonable encryption scheme. Only in the second phase of the experiment is this key revealed. At this point we have to reveal the decryption key to the adversary who is attacking the $\QFE$ construction.

This part of the proof is reminiscent of the transformation given in \cite{TCC:AnaKal21} who also use the mode change via a flag bit trick. They use classical functional encryption to transform secret-key unclonable encryption into public-key unclonable encryption. \footnote{They also explain very well why a normal public-key encryption scheme is not sufficient but a functional encryption scheme is.} 

In their construction  it is possible to hardcode the classical decryption key of an unclonable encryption scheme into the circuit description and then create a function secret key for this circuit to complete the proof. Unfortunately we cannot directly apply same technique as \cite{TCC:AnaKal21}. In our case the decryption key is a quantum state and our QFE scheme does not support hardcoding quantum states into the circuit description.  

To solve this issue we create 2n EPR pairs and put one qubit of each EPR pair in the ciphertext. Later we can teleport the quantum decryption keys into the ciphertext and hardcode the correction keys of the teleportation into the function secret key. The circuit applies the correction keys to the decryption key and can then use it to decrypt the $\rho_{UE}$ ciphertext. Hardcoding the teleportation keys into the circuit introduces a randomization of the function secret key which is why we do not achieve fully fledged unclonable encryption but only a version with variable decryption keys. Furthermore we have to make sure that each part of the reduction $B$ and $C$ who each obtain a quantum decryption key $|\dk\rangle$ can create a valid decryption key for their part of the reduction. Since the EPR pairs for the teleportation procedure cannot be held by both $B$ and $C$ at the same time we need to provide two teleportation slots. Then $B$ and $C$ each individually teleport the decryption key into the ciphertext and create a function secret key based on their teleportation keys. The teleportation keys $(a,b)$ are uniformly random bits, so the function secret keys that depend on them are indistinguishable from regular function secret keys that were created with freshly sampled bits.

In the second step of the proof we construct a ciphertext with the the flag bit set to 1  to reduce multi-bit security of our unclonable functional encryption scheme to the single bit security of the unclonable encryption scheme of \cite{AKY24}.

\subsection{QMIFE and Applications to Quantum Indistinguishability Obfuscation}\label{QMIFE:subsec:intro}
In the classical setting much research has focused on improving on the trade-off inherent between the size of allowable circuits and the length of the ciphertext. Recall the schemes constructed in \cite{CCS:SahSey10, C:GorVaiWee12}, as well as our scheme given in \cref{qfe:sec}, are considered non-succinct as the size of the ciphertext must be at least as large as the circuit description of allowable circuits. In \cite{EC:GGGJKH14}, a stronger variant on FE, known as multi-input functional encryption (MIFE) is introduced. In \cite{EC:GGGJKH14}  it is shown that MIFE  enables applications towards indistinguishability obfuscation without the requirement of succinctness.

\subsubsection{MIFE} Multi-Input Functional Encryption (MIFE) extends traditional functional encryption to handle functions over multiple ciphertexts, potentially encrypted under different keys. This general framework allows for the computation of aggregate information from various data sources, going beyond single-input functional encryption. In MIFE, the owner of a master secret key (MSK) can derive special function keys that enable the evaluation of an n-ary function $f(x_1, \dots, x_n)$ on ciphertexts corresponding to different messages, even when encrypted by different parties. Such multi-input functionality has been shown to allow for many powerful applications such as multi-party delegated computation, and construction of indistinguishability obfuscation (iO) and virtual black-box obfuscation (VBBO).

\subsubsection{QMIFE}
Analogously, a quantum multi-input functional encryption (QMIFE) scheme is a QFE scheme that can evaluate a function on multiple, individually encrypted quantum inputs. In our definition of QMIFE we switch to the secret-key version of functional encryption. Therefore the ciphertexts cannot be encrypted by anyone but only by the holder of encryption secret keys. Additionally the scheme is tagged with an encryption limit $k$ which indicates how many ciphertexts per encryption key can be obtained. 

The IND-definition for QMIFE readily generalizes using methods from the IND-security definition for QFE: For any combination of inputs and circuit queries the restriction of admissible queries has to be fulfilled. In the SIM-security definition a new uniquely quantum challenge arises. Informally we want to give the simulator exactly the information that we want to allow a participant in the QMIFE scheme to learn. In the classical setting this corresponds to the output of the quantum circuit for any combination of challenge inputs. In the quantum setting we have the problem that different combinations of inputs are possible but the quantum ciphertexts are not necessarily reusable. If we give the simulator access to all possible circuit outputs we are giving him too much information since obtaining all outputs might not be a physical process. On the other hand, allowing the simulator to obtain only one output is too little information.

For instance, an adversary could attempt to run decryption $\Dec(\sk_C, \rho_{\ct_1}, \rho_{\ct_2})$ on two ciphertext registers, measure one register, uncompute the decryption, and then swap the first register with a new state. We solve this issue by giving the simulator access to a trusted party that holds the input messages. The simulator can query the trusted party by defining a circuit and indices to select the input messages. Then the trusted party carries out the circuit evaluation, moves the output into a new register by applying a CNOT gate to every qubit and uncomputes the circuit on the input registers. The trusted party returns the output to the adversary and proceeds in the same manner for additional queries. Now the state that is obtained by the simulator is entangled with the trusted party and any measurements that might be performed by the simulator disturb the state and influence future circuit evaluations. This simulates the information we expect a recipient of a number of ciphertexts and function keys to be able to compute without breaking the security of the QMIFE scheme.

Our formal presentation of QMIFE, including both IND-security and SIM-security definitions, is given in \cref{qfe:sec:qmife}. Additionally, our treatment of QMIFE covers multi-query QFE as a special case. Our main application is given in \cref{qfe:thm:indtoqio} and \cref{QMIFE:thm:VBBO} which provide the following quantum analogue of the celebrated reductions to iO and VBBO given in \cite{EC:GGGJKH14}.

\begin{theorem}[Informal] 
Any single-query non-adaptive IND-secure QMIFE unconditionally implies qiO.
\end{theorem}

\begin{theorem}[Informal]
 Any single-query non-adaptive SIM-secure QMIFE scheme unconditionally implies virtual black box quantum obfuscation.   
\end{theorem}

\onlyeprint
\subsection{Additional Related Work}
\subsubsection{Functional Encryption}
While we are the first to consider functional encryption for quantum circuits there has been a series of works enhancing classical functional encryption using quantum techniques. By adding the possibility to certifiably delete the ciphertext of the FE scheme \cite{EC:HKMNPY24} construct certified everlasting functional encryption.
\cite{AC:KitNis22} define and create functional encryption with secure key leasing from any secret-key FE and they construct FE with single decryptor against bounded collusions assuming sub-exponentially secure indistinguishability obfuscation and the sub-exponential hardness of the learning with errors (LWE) problem.
Using different techniques \textcite{CakGoy23} construct functional encryption with copy protected secret keys against unbounded collusions from sub-exponentially secure indistinguishability obfuscation, one-way functions and LWE.

\subsubsection{Unclonable Encryption}
The notion of unclonable encryption was formally defined by \cite{BL20}, previously a similar notion was introduced by \cite{Got03}. Since then the gold standard of indistinguish\-able-unclonable secure encryption with negligible adversarial advantage has only been achieved in the quantum random oracle model by~\cite{C:AKLLZ22} and a construction in the plain model remains an open question. Various alternative notions of unclonable encryption have been achieved such as device-independent unclonable encryption with variable secret keys~\cite{KT20}, unclonable encryption with interaction~\cite{BC23}, unclonable encryption with quantum decryption keys~\cite{AKY24}. The encryption procedure of~\cite{KT20} is defined as an interactive process, although the interaction can be removed in the trusted device setting. The relationship of unclonable encryption to other primitives that require a form of unclonability such as quantum money \cite{wiesnercoding, STOC:AarChr12} and copy protected programs \cite{Aaronson09, TCC:AnaKal21, TCC:BJLPS21, C:AKLLZ22, CMP24, STOC:ColGun24} has also been studied. 

We note that the~\cite{KT20} scheme is only proven to achieve the weaker notion of unclonable security, where a successful attack requires both parties to guess the entire message. In our work all security notions obey the stronger uncloneable-indisinguishability. For brevity we often simply refer to this as unclonable security in the text.

\subsubsection{Quantum Obfuscation} \textcite{AF16} provide a quantum analouge of the classical impossibility result for virtual black box obfuscation (VBB), showing that the notion of quantum virtual black box obfuscation (QVBB) is also impossible to achieve. Furthermore, \cite{C:ABDS21} show that a quantum scheme cannot achieve VBB for classical circuits either. The first feasability result for qiO was obtained by \textcite{LC:BroKaz21} for circuits with log-many non-clifford gates relying on classical iO. Since then several works have put forth candidate constructions using a wide variety of techniques. \textcite{ITCS:BarMal22} construct qiO for null circuits\footnote{Null quantum circuits are circuits that reject on every input with overwhelming probability.} assuming classical VBB. \cite{STOC:BKNY23} construct QVBB for pseudo-deterministic circuits\footnote{A pseudo-deterministic circuit takes as input a classical string and outputs a deterministic bit with overwhelming probability taken over the randomness introduced by the quantum circuit.} with a classical description assuming classical VBB. \cite{STOC:BBV24} improve upon this result by constructing ideal QVBB\footnote{Ideal QVBB is very similar but slightly stronger than QVBB.} for pseudo-deterministic circuits with a quantum description assuming classical VBB. Since classical VBB is known to be impossible these constructions are only candidates for qiO meaning we can hope that if the classical VBB is instantiated with classical iO the constructions can be proven secure with new ideas.

\subsection{Open Questions}
An important open question is the construction of quantum indistinguishability obfuscation. In this work we make a step towards exploring the relationship of quantum functional encryption to qiO via multi-input quantum functional encryption. It is an interesting open question if QMIFE can be constructed by for example leveraging classical multi-input functional encryption which can be constructed from classical iO, a reasonable assumption in the construction of qiO. 

Another open question that this work raises are enhanced versions of quantum functional encryption. A QFE scheme with succinct ciphertext would have interesting applications such as delegated computation~\cite{STOC:GKPVZ13} and can potentially provide another route towards qiO. In the classical setting techniques to transform succinct FE to iO haven been explored extensively~\cite{FOCS:BitVai15, AJS15, GGHRSW16, C:AnaJai15} and might be applicable in the quantum setting to. 

Lastly we only construct QFE for a single key query and leave it as an open problem to construct QFE secure against multiple key queries.

\subsubsection{Acknowledgements} We would like to thank Henry Yuen for helpful discussions. Arthur Mehta is supported by NSERC Alliance Consortia Quantum grants, reference number:
ALLRP 578455 - 22 and the NSERC Discovery Grants Program. 
\section{Preliminaries}

For an integer $n \in \N$ we write $[n] = \{1, \ldots, n\}$. Let $p(\cdot)$ denote a polynomial. Let $negl(\cdot)$ denote a negligible function $f$, i.e. for every constant $c \in \N$ there exists a positive integer $n_0$ such that for all $n>n_0$, $f(n) < n^{-c}$.

Let $\Hilb_n$ denote a finite dimensional Hilbert space of dimension $2^n$ and let a pure quantum state be denoted by a vector $|\psi\rangle \in \Hilb$. Let a mixed quantum state be denoted as $\rho \in D(\Hilb_n)$ where $D(\Hilb_n)$ is the set of density operators on $\Hilb_n$ which are positive semidefinite and have trace equal to 1. A general quantum operation is a completely positive trace preserving (CPTP) map $\Phi: D(\Hilb_n)\rightarrow D(\Hilb_m)$.

For a classical string $x \in \bin^n$ we let $|x| = n$ denote the length of the string and for a quantum state $\rho \in D(\Hilb_n)$ we let $|\rho| = n$ denote the size, i.e. the number of qubits.  

Let $\Tr$ denote the trace operator. Let the partial trace be denoted as $\Tr_{(b)}[\rho_{ab}] = \rho_a = \Tr(\rho_b) \rho_a$. We write $\rho_{x_i}$ to denote taking the partial trace of everything but the i-th qubit $\Tr_{(\Bar{i})}(\rho_x) = \rho_{x_i}$. We write $\rho^A$ to denote that the qubits in $\rho$ are conceptually grouped together in a register $A$.

A family of quantum circuits $\{C_\lambda\}_{\lambda \in \N}$ is called uniform if there exists a deterministic Turing machine running in time poly($\lambda$) such that on input $1^\lambda$ it outputs a description of $C_\lambda$.
A quantum polynomial time (QPT) algorithm is a polynomial-time uniform family of quantum circuits.

A universal gate set for quantum circuits is the Clifford group consisting of the controlled-not gate $\mathsf{CNOT}$, phase gate $\mathsf{P}$ and Hadamard gate $\mathsf{H}$ with additionally the T-gate $\mathsf{T}$. Let $\Xgate$ and $\Zgate$ be the following gates
$$ \Xgate = \left(\begin{matrix}
    0 & 1 \\
    1 & 0 \\
\end{matrix}\right) \quad \Zgate = \left(\begin{matrix}
    1 & 0 \\
    0 & -1 \\
\end{matrix}\right)$$

\subsection{Indistinguishability of Quantum States}
The \textit{trace distance} between two quantum states $\rho, \sigma \in \Dist(\Hilb_n)$ is defined as  
$$\T(\rho, \sigma ) = \frac{1}{2} \Tr\left( \sqrt{(\rho-\sigma )^\dag(\rho- \sigma )}  \right)$$.

Let $\mathcal{R} = \{\rho_n\}_{n\in\N}$ and $\mathcal{S} = \{\sigma_n\}_{s\in\N}$ be two ensembles of quantum states such that $\rho_n$ and $\sigma_n$ are $n$-qubit states. 
$\mathcal{R}$ and $\mathcal{S}$ are called \textit{perfectly indistinguishable} if for all $n$: $\rho_n = \sigma_n$.

$\mathcal{R}$ and $\mathcal{S}$ are called \textit{satistically indistinguishable} if there exits a negligible function negl such that for all sufficiently large $n$: $$\T(\rho_n,\sigma_n) \leq negl(n)$$.

$\mathcal{R}$ and $\mathcal{S}$ are called \textit{computationally indistinguishable} if there exits a negligible function negl such that for all QPT distinguisher $\Dist$ and all states $\rho_n \in \mathcal{R}$ and $\sigma_n \in \mathcal{S}$:
 $$ \left| \Pr[\Dist(\rho_n) = 1] - \Pr[ \Dist(\sigma_n) = 1] \right| \leq negl(n)$$.

The diamond norm for two quantum channels $\Phi$ and $\Psi$ mapping a n-qubit quantum state to a m-qubit quantum state is defined as follows: 
$$|| \Phi - \Psi||_\diamond = \max_{\rho \in D(\Hilb^{2n})} \T((\Phi \otimes I)\rho - (\Psi \otimes I)\rho )$$

\subsection{Quantum Randomized Encodings}
We recall the following definitions from \cite{BY20}.

\paragraph{Classical Description of Quantum Circuits}
A quantum circuit is a tuple $(\mathcal{P}, \mathcal{G})$ where $\mathcal{P}$ is the topology of the circuit and $\mathcal{G}$ is a set of unitaries. 
The topology of a quantum circuit is a tuple $(\mathcal{B}, \mathcal{I}, \mathcal{O}, \mathcal{W}, \texttt{inwire}, \texttt{outwire}, \mathcal{Z}, \mathcal{T})$.
\begin{enumerate}[align=left, leftmargin=2.8em]
    \item $\mathcal{I}$ is an ordered set of input terminals.
    \item $\mathcal{Z}$ is a subset of $\mathcal{O}$ which indicates ancilla qubits that are to be initialised to the state $|0\rangle$.
    \item $\mathcal{O}$ is an ordered set of output terminals.
    \item $\mathcal{T}$ is the set of output terminals to be traced out. 
    \item $\mathcal{W}$ is the set of wires. 
    \item $\mathcal{B}$ are placeholder gates. For every $g \in \mathcal{B}$ $\texttt{inwire(g)}$ describes an ordering of input wires $w \in \mathcal{W}$ and $\texttt{outwire(g)}$ describes an ordering of output wires $w \in \mathcal{W}$.  For every $g \in \mathcal{B}$ the number of input and output wires is equal.
    \item The disjoint sets $\mathcal{I}, \mathcal{O} ,\mathcal{B}$ form the nodes of the circuit. Together with the set $\mathcal{W}$ as edges they define a directed acyclic graph. 
\end{enumerate}

The gate set $\mathcal{G}$ defines a unitary of the appropriate size for every node in $\mathcal{B}$. The evaluation of a circuit $C= (\mathcal{P}, \mathcal{G})$ on state $\rho$ of size $|\mathcal{I}|$ is defined as $C(\rho, |0\rangle^{\otimes |\mathcal{Z}|}) = \sigma$ where $\sigma$ resulted from applying the gates in $\mathcal{G}$ according to the topology and tracing out the qubits specified by $\mathcal{T}$. The size of a quantum circuit is the number of wires in $\mathcal{W}$. The descritpion of quantum operations by a quantum circuit describes a CPTP map.

\begin{definition}{Quantum Randomized Encodings (QRE)}

    Let $(\Encode, \Decode, \Sim)$ be QPT algorithms. Let $C$ denote a class of general quantum circuits.
\begin{itemize}[align=left]
    \item[$\Encode(F,\rho_x,r,\rho_e) \rightarrow \hat{F}(\rho_x,r)$:] Encode$(F,\rho_x,r,\rho_e)$ takes a function $F \in C$, quantum input $\rho_x$, classical randomness r and a set of EPR pairs $\rho_e$ and outputs a quantum randomized encoding $\hat{F}(\rho_x,r)$.
    \item[$\Decode(\hat{F}(\rho_x,r), T) \rightarrow F(\rho_x)$:] Decode takes as input a quantum randomized encoding $\hat{F}(\rho_x,r)$ and the topology $T$ of the function F and outputs $F(\rho_x)$. 
    \item[$\Sim(F(\rho_x), T)$:] Sim takes as input the value $F(\rho_x)$ and the topology of F and simulates a quantum randomized encoding.
\end{itemize}
A QRE scheme fulfills the following properties:
\begin{itemize}
    \item \textbf{Correctness} For all quantum states $(\rho_x,\rho_z)$ and randomness r it holds that 
    $$(\Decode(\hat{F}(\rho_x,r),T), \rho_z)= (F(\rho_x),\rho_z)$$.
    \item \textbf{(t,$\epsilon$)-Privacy} For all quantum states $(\rho_x,\rho_z)$ and distinguishers of size $t$ it holds that 
    $$(\Sim(F(\rho_x)),\rho_z)\approx_{\epsilon} (\hat{F}(\rho_x,r),\rho_z)$$.
\end{itemize}
\end{definition}

A QRE can additionally fulfill the following property:
\begin{definition} Decomposability
    \begin{itemize}
    \item \textbf{Decomposability:} The encoding $\hat{F}$ is decomposable if there exists an operation $\hat{F}_{off}$ (called the offline part of the encoding) and a collection of input encoding operations $\hat{F}_1, \dots, \hat{F}_n$ such that for all inputs $\rho_x = (\rho_{x_1},\dots , \rho_{x_n})$,  $\hat{F}(\rho_x,r) =  (\hat{F}_{off}, \hat{F}_1, \dots, \hat{F}_n)(r\rho_x,r,\rho_e)$ where the functions $\hat{F}_{off}, \hat{F}_1, \dots, \hat{F}_n$ act on disjoint subsets of qubits from $\rho_e$, $\rho_x$ (but can depend on all bits of r), each $\hat{F}_i$ acts on a single qubit $\rho_{x_i}$, and $\hat{F}$ does not act on any of the qubits of $\rho_x$.
    \item \textbf{Classical Encoding of Classical Inputs:} If an input qubit $x_i$ is classical, then the input encoding operation $\hat{F}_i$ is computable by a classical circuit.
\end{itemize}
\end{definition}

\begin{definition}Quantum Garbled Circuits (QGC)

    Quantum Garbled Circuits are an instantiation of QRE that fulfill the Decomposability property with classical encodings of classical inputs. For a quantum circuit of size $s$ the randomized encoding can be computed by a circuit of size $poly(\lambda,s)$ and fulfills computational security, that is for every polynomial $t(\lambda)$ there exists a negligible function  $\eps = negl(\lambda)$ such that the scheme is $(t', \epsilon')-private$, where $t'(\lambda) = t(\lambda) - poly(\lambda,s)$ and $\epsilon'(\lambda) = \epsilon(\lambda) \cdot s$. 
    The decoding and simulation procedures are computable in time $poly(\lambda)\cdot s$.
\end{definition}

Additional preliminaries regarding classical functional encryption, quantum obfuscation and unclonable encryption can be found in \cref{qfe:addprelims}.
\label{qfe:addprelims}

\subsection{The Quantum One Time Pad}
The Quantum One Time Pad (QOTP) ~\cite{AMTW00} is the quantum analogue to the classical One Time Pad. 
\begin{definition}(Quantum One Time Pad)
\begin{itemize}[align=left, leftmargin=2.8em]
    \item[$\bm{\Enc(\sk, |\phi\rangle \in \Hilb_1) \rightarrow |\ct\rangle}$] Given a secret key $\sk = (a,b)$ and a quantum message $|\psi\rangle$ apply the following operation to the state to obtain the ciphertext:
    $$ |\ct\rangle  = \Xgate^a \Zgate^b |\phi\rangle$$
    \item[$\bm{\Dec(\sk, |\ct\rangle) \rightarrow |\phi\rangle}$] Given a secret key $\sk = (a,b)$ and a ciphertext apply the following operation to obtain the message:
    $$ |\phi\rangle  = \Xgate^a \Zgate^b |\ct\rangle$$

\end{itemize}
\end{definition}
When the key $\sk = (a,b)$ is chosen uniformly at random from $\bin^2$, the QOTP information theoretically hides the state. The technique generalises to mulit-qubit states by encrypting qubit by qubit. 

\subsection{Quantum State Teleportation}

Two spatially separated parties $A$ and $B$ can teleport a quantum state from one person to the other by using shared entanglement and classical communication~\cite{Teleport}. $A$ holds the state $\rho$ and one qubit of an EPR pair, $B$ holds the other qubit of the EPR pair. $A$ performs a Bell measurement on the two states and obtains the correction keys $(a,b)$. The keys $(a,b)$ are send to $B$ who applies an $X$ gate to the state if $a=1$ and a Z gate to the state if $b=1$. Now Bob holds the state $\rho$. The technique generalises to mulit-qubit states by teleporting qubit by qubit.

\subsection{Classical Functional Encryption}
\begin{definition}(Functional Encryption)
Let $\lambda \in \N$ be the security parameter. Let $\mathcal{F} = \{\mathcal{F_\lambda}\}_{\lambda \in \N}$ be a class of circuits with input space $\mathcal{X} = \{\mathcal{X_\lambda}\}_{\lambda \in \N}$ and output space $\mathcal{Y} = \{\mathcal{Y_\lambda}\}_{\lambda \in \N}$. A functional encryption scheme is defined by the PPT algorithms $\FE = (\Setup, \Keygen,\allowbreak \Enc, \Dec)$. 
    \label{def:fe}
    \begin{description}
    \item[Setup$(\secparam) \rightarrow (\mpk,\msk)$:] given the security parameter $\secparam$ outputs the master public key $\mpk$ and the master secret key $\msk$.
    \item[KeyGen($\msk,f) \rightarrow \sk_{f}$:] given the master secret key $\msk$ and a circuit $f$ and outputs a  function key $\sk_{f}$.
    \item[Enc($\mpk,\msg) \rightarrow \ct$:] given $\mpk$ and a message $m \in \mathcal{X}$ output the ciphertext $\ct$.
    \item[Dec$(\sk_{f},\ct ) \rightarrow y$:] given a ciphertext $\ct$ and $\sk_{f}$ output a value $y \in \mathcal{Y}$.
    \end{description}
\end{definition}

The scheme has to fulfill the following correctness and security properties:

\begin{definition}(Correctness)
Let $(\mpk,\msk) \leftarrow$ $\Setup(\secparam)$, $\sk_{f} \leftarrow$ $\Keygen(\msk$, $f), \ct \leftarrow$ $\Enc(\mpk,\msg)$. Then the FE is correct, if for all $f \in \mathcal{F}$ and $\msg \in \mathcal{X}$ it holds that $f(\msg) = \Dec(\sk_{f},\ct)$.

\end{definition}

\begin{definition}(Single-Query IND-Security for Classical Functional Encryption)
Let $\lambda \in \N$ be the security parameter and let $\Adv$ be a QPT adversary.
Consider the experiment $\exp_{\Adv,b}^{\FE}(\secparam)$:
\begin{enumerate}[align=left,leftmargin=2.8em]
    \item $\FE.\Setup(\secparam) \rightarrow (\mpk,\msk)$
    \item $(\msg_0, \msg_1, \st) \leftarrow \Adv^{\sk_f \leftarrow\Keygen (\msk,\cdot)}(1^\lambda, \mpk)$ where $\msg_0, \msg_1$ have to be admissible queries for a function $f$ that $\Adv$ queries, they fulfil $f(\msg_0) = f(\msg_1)$.
    \item Sample $b \leftarrow \bin$
    \item $\ct \leftarrow \Enc(\mpk,\msg_b)$.
    \item $b' \leftarrow \Adv^{O(\cdot)}(1^\lambda, \ct,\st)$.
    \item If $b'=b$ the adversary wins and the experiment outputs $1$. Otherwise, the experiment outputs 0. 
\end{enumerate}
A functional encryption scheme is said to have single-key IND-security if for all QPT adversaries $\Adv$, there exists a negligible function $negl$ such that for all $\secpar \in \mathbb{N}$:
$$\left|\Pr\left[1 \leftarrow \Exp_{\Adv, b=0}^{Ind}\right] - \Pr\left[1 \leftarrow \Exp_{\Adv, b=1}^{Ind}\right]\right| \leq negl(\lambda)$$
where the random coins are taken over the randomnes of $\Adv$, $\Setup, \Keygen$ and $\Enc$.

\textbf{Adaptive vs. Non-adaptive security}
\begin{itemize}
    \item The scheme is called non-adaptively secure if the the adversary only queries the $\Keygen$ oracle before receiving a ciphertext. Then the oracle $O(\cdot)$ is the empty oracle. 
    \item The scheme is called adaptively secure if the adversary can either query the $\Keygen$ oracle before or after receiving the ciphertext. Then the oracle $O(\cdot)$ is the function $\Keygen(\msk, \cdot)$.
\end{itemize}
\end{definition}

\begin{definition}(Single-Query SIM-security)
    Let $\lambda$ be the security parameter and let $\Adv = (\Adv_1, \Adv_2)$ be a QPT adversary and let $\Sim$ be a QPT simulator. 
   \begin{table}[H]
        \centering
        \begin{tabular}{p{6cm}|p{6cm}}
            $\Exp_{\Adv}^{Real}(1^\lambda)$ & $\Exp_{\Adv}^{Ideal}(1^\lambda)$ \\
            $(\mpk,\msk) \leftarrow \Setup(1^\lambda)$&$(\mpk,\msk) \leftarrow \Setup(1^\lambda)$ \\
            $ (\msg, \st) \leftarrow \Adv_1^{O_1(\cdot)}(1^\lambda, \mpk)$ & $ (\msg,\st) \leftarrow \Adv_1^{O_1(\cdot)}(1^\lambda, \mpk)$ \\
            ${\ct} \leftarrow \Enc(\mpk, \msg)$& ${ct} \leftarrow \Sim(1^\lambda, \mpk, \mathcal{V})$ \\
            & \quad where $\mathcal{V} = (C, \sk_{C}, C(\msg), 1^{|\msg|})$ if $\Adv$\\ & \quad queried $O_1$ on $C$ and $\mathcal{V} = \emptyset$ otherwise.\\
            $\alpha \leftarrow \Adv_2^{O_2(\cdot)}({\ct}, \st)$ & $\alpha \leftarrow \Adv_2^{O_2'(\cdot)}(\ct, \st)$\\
            The experiment outputs the state $\alpha$ & The experiment outputs the state $\alpha$\\
        \end{tabular}
    \end{table}
The FE scheme is single-query simulation-secure if for any adversary $\Adv$ and all messages $\msg$ there exists a simulator $\Sim$ such that the real and ideal distributions are computationally indistinguishable:
$$ \{ \Exp_{\Adv}^{Real}(1^\lambda)\}_{\lambda \in \N} \approx_c \{ \Exp_{\Adv}^{Ideal}(1^\lambda)\}_{\lambda \in \N} $$.

\textbf{Adaptive vs Non-adaptive security:}
\begin{enumerate}[align=left,leftmargin=2.8em]
    \item Non-adaptive: the adversary $\Adv_1$ is allowed to make one key query to $O_1(\cdot)$ where the oracle $O_1(\cdot)$ is $\Keygen(\msk, C) \rightarrow sk_C$.
    \item Adaptive: the adversary is allowed to make one key query either to $O_1(\cdot)$ or $O_2(\cdot)$ ($O'_2(\cdot)$ in the ideal world) where $O_1(\cdot)$ and $O_2(\cdot)$ are $\Keygen(\msk, C) \rightarrow sk_C$ and $O_2'(\cdot)$ is a $\Keygen$ oracle controlled by the simulator $\sk_C \leftarrow \Sim(1^\lambda, \msk, C, C(\msg), 1^{|\msg|})$. The simulator is stateful, in this invocation $\Sim$ has access to the state of the simulator from it's first invocation where it produced the ciphertext. 
\end{enumerate}
\end{definition}

In this work we only require a very simple functional encryption schemes: We require a single-query adaptively SIM-secure FE scheme for the identity circuit and we require a single-query adaptively SIM-secure FE scheme for a family of two cicruits. Such schemes are constructed in \cite{C:GorVaiWee12}.

\subsection{Classical Multi-input Functional Encryption}
We recall the syntax and security definition of a classical multi-input functional encryption scheme (MIFE)~\cite{EC:GGGJKH14}. We only consider the case where all encryption keys are secret.

Let $\mathcal{X}=\left\{\mathcal{X}_\lambda\right\}_{\lambda \in \mathbb{N}}$ and $\mathcal{Y}=$ $\left\{\mathcal{Y}_\lambda\right\}_{\lambda \in \mathbb{N}}$ be ensembles where each $\mathcal{X}_\lambda$ and $\mathcal{Y}_\lambda$ is a finite set. Let $\mathcal{F}=\left\{\mathcal{F}_\lambda\right\}_{\lambda \in \mathbb{N}}$ be an ensemble where each $\mathcal{F}_\lambda$ is a finite collection of $n$-ary functions. Each function $f \in \mathcal{F}_\lambda$ takes as input $n$ strings $x_1, \ldots, x_n$, where each $x_i \in \mathcal{X}_\lambda$ and outputs $f\left(x_1, \ldots, x_n\right) \in \mathcal{Y}_\lambda$. A multi-input functional encryption scheme is additionally parametrized by a parameter $k$ which denotes how many ciphertexts can be produced for one encryption key $\ek$.

A multi-input functional encryption scheme $\MIFE$ for $\mathcal{F}$ consists of four algorithms $(\Setup, \Keygen,$ $ \Enc, \Dec)$ as described below.
\begin{itemize} [align=left,leftmargin=2.8em]
    \item[$\bm{\Setup}$]  $\Setup(1^\lambda, n) \rightarrow (\msk, \ek_1, \ldots, \ek_n)$ is a PPT algorithm that takes as input the security parameter $\lambda$ and the function arity $n$. It outputs $n$ encryption keys $\ek_1, \ldots, \ek_n$ and a master secret key $\msk$.
   
    \item[$\bm{\Keygen}$] $\Keygen(\msk, $f$)\rightarrow \sk_f$ is a PPT algorithm that takes as input the master secret key $\msk$ and an $n$-ary function $f \in \mathcal{F}_\lambda$ and outputs a corresponding secret key $\sk_f$.
    \item[$\bm{\Enc}$] $\Enc(\ek,x) \rightarrow \ct$ is a PPT algorithm that takes as input an encryption key $\ek_i \in\left(\ek_1, \ldots, \ek_n\right)$ and an input message $x \in \mathcal{X}_\lambda$ and outputs a ciphertext $\ct$. In the case where all of the encryption keys $\ek_i$ are the same, we assume that each ciphertext $\ct$ has an associated label $i$ to denote that the encrypted plaintext constitutes an $i$'th input to a function $f \in \mathcal{F}_\lambda$. For convenience of notation, we omit the labels from the explicit description of the ciphertexts. In particular, note that when $\ek_i$'s are distinct, the index of the encryption key $\ek_i$ used to compute $\ct$ implicitly denotes that the plaintext encrypted in $\ct$ constitutes an $i$'th input to $f$, and thus no explicit label is necessary.
    \item[$\bm{\Dec}$] $\Dec(\sk_f, \ct_1, \ldots, \ct_n)\rightarrow y$ is a deterministic algorithm that takes as input a secret key $\sk_f$ and $n$ ciphertexts $\ct_1, \ldots, \ct_n$ and outputs a string $y \in \mathcal{Y}_\lambda$.
\end{itemize}

\begin{definition}(Correctness)
    A multi-input functional encryption scheme $\mathcal{F E}$ for $\mathcal{F}$ is correct if for all $f \in \mathcal{F}_\lambda$ and all $\left(x_1, \ldots, x_n\right) \in \mathcal{X}_\lambda^n$ :

    $$\mathrm{Pr}\left[\begin{array}{l} \Dec(\sk_f, \Enc(\ek_1, x_1), \ldots, \Enc(\ek_n,x_n)) = f(x_1,\ldots, x_n) :\\
    (\msk, \ek_1, \ldots, \ek_n)\leftarrow\Setup(1^\lambda,n), \sk_f \leftarrow \Keygen(\msk, f)  
    \end{array}\right]= 1- negl(\lambda)$$
where the probability is taken over the coins of $\Keygen, \Setup, \Enc$.
\end{definition}

\begin{definition}(Compatibility of function and message queries)
\label{qfe:def:compatability}

Let $\{f\}$ be any set of n-ary functions $f \in \mathcal{F_\lambda}$. Let $X^0, X^1$ a pair of input vectors where $X^b = \{x^b_{1,j}, \dots, x^b_{n,j}\}_{j \in [k]}$. We say $(X^0,X^1)$ and $\{f\}$ are compatible if they satisfy the following property:

For every $f \in \{f\}$ and every $j_1, \dots, j_n \in [k]$

$$ f(x^0_{1,j_1}, \dots, x^0_{n,j_n}) = f(x^1_{1,j_1}, \dots, x^1_{n,j_n})$$
\end{definition}

\begin{definition}(Classical MIFE selective IND-Security)

    A multi-input functional encryption scheme $\MIFE$ for $n$-ary functions $\mathcal{F}$ is $k$-IND-secure if for every PPT adversary $\mathcal{A}=\left( \mathcal{A}_1, \mathcal{A}_2\right)$, the advantage of $\mathcal{A}$ defined as
$$
\operatorname{Adv}_{\mathcal{A}}^{\MIFE, \mathrm{IND}}\left(1^\lambda\right)=\left|\operatorname{Pr}\left[\Exp_{\mathcal{A}}^{\operatorname{IND}-\MIFE}\left(1^\lambda\right)=1\right]-\frac{1}{2}\right| \leq negl(\lambda)
$$
 where:
$$
\begin{aligned}
& \Exp_{\mathcal{A}}^{\operatorname{IND}-\MIFE}\left(1^\lambda\right): \\
& \left(\vec{X}^0, \vec{X}^1, \st_1\right) \leftarrow \mathcal{A}_1\left(1^\lambda, n \right) \text { where } \vec{X}^{\ell}=\left\{x_{1, j}^{\ell}, \ldots, x_{n, j}^{\ell}\right\}_{j \in [k]} \\
&( \{\ek_i\}_{i\in[n]}, \msk) \leftarrow \Setup\left(1^\lambda, n\right) \\
& b \leftarrow\{0,1\} \\
& \ct_{i, j} \leftarrow \Enc\left(\mathrm{ek}_i, x_{i, j}^b\right) \quad \forall i \in[n], j \in[k] \\
& b^{\prime} \leftarrow \mathcal{A}_2^{\Keygen(\msk, \cdot)}\left(\st_1, \{\ct_{i,j}\}_{i \in [n], j \in [k]}\right) \\
& \text{Output: } \left(b=b^{\prime}\right)
\end{aligned}
$$

Let $\{f\}$ denote the entire set of key queries made by $\mathcal{A}$ at any point dirung the game. Then, the challenge message vectors $\vec{X}_0$ and $\vec{X}_1$ chosen by $\mathcal{A}$ must be compatible with $\{f\}$ (\cref{qfe:def:compatability}).
\end{definition}

\begin{lemma}\cite{EC:GGGJKH14}
    Let  $k = k(\lambda)$ be a fixed $poly(\lambda)$. Then, assuming indistinguishability obfuscation for all polynomial-time computable classical circuits and one-way functions, there exists a $k-\MIFE$ scheme that is selectively IND-secure.
\end{lemma}

\begin{definition}(Classical MIFE Sim-Security)
    A multi-input functional encryption scheme for n-ary functions is k-SIM-secure if for every QPT adversary $\Adv = (\Adv_1, \Adv_2)$ there exists a stateful simulator $\Sim$ such that the outputs of the following experiments are computationally indistinguishable: 

     \begin{table}[H]
        \centering
        \begin{tabular}{p{6cm}|p{6cm}}
            $\Exp_{\Adv}^{Real}(1^\lambda)$ & $\Exp_{\Adv}^{Ideal}(1^\lambda)$ \\
            $(\{\ek_i\}_{i\in[n]},\msk) \leftarrow \Setup(1^\lambda,n)$& \\
            $ (X,  \st) \leftarrow \Adv_1^{\Keygen(\msk,\cdot)}(1^\lambda,n)$ & $(X,  \st) \leftarrow \Adv_1^{O_1(\cdot)}(1^\lambda)$\\
             where $X = \{ {\msg_{1,j}}, \dots, {\msg_{n,j}} \}_{j \in [k]}$&where $X = \{ {\msg_{1,j}}, \dots, {\msg_{n,j}} \}_{j \in [k]}$\\
            $ {\ct_{i,j}} \leftarrow \Enc(\ek_i,  {\msg_{i,j}}) \quad \forall i \in [n], j \in [k]$&$ \{\ct_{i,j}\}_{i\in [n],j\in[k]} \leftarrow \Sim^{\TP(\cdot)}(1^\lambda,1^{|m_{i,j}|}) $  \\
            $\alpha \leftarrow \Adv_2^{\Keygen(\msk,\cdot)}( \{\ct_{i,j}\}_{i\in[n],j\in[k]},\st)$ & $\alpha \leftarrow \Adv_2^{O_2(\cdot)}( \{\ct_{i,j}\}_{i\in[n],j\in[k]},\st)$\\
            The experiment outputs $\alpha$ & The experiment outputs $\alpha$\\
        \end{tabular}
    \end{table}
    
    where the oracle $\TP(\cdot)$ denotes the ideal world trusted party. $\TP$ accepts queries of the form $(g,(j_1, \dots, j_n))$ and outputs $g(\msg_{1,j_1}, \dots, \msg_{n,j_n})$.
    
    $O_1(\cdot)$ is a $\Keygen$ oracle controlled by the simulator and $O_2(\cdot)$ is a $\Keygen$ oracle controlled by the simulator with access to $\TP$. We say $\Sim$ is admissible if $\Sim$ only queries $\TP$ on functions that $\Adv$ queried to its oracle. 
    
    In a single-query secure scheme $\Adv$ (in the real world) can only make a single query to the $\Keygen$ oracle or (in the ideal world) a single query to either $O_1(\cdot)$ or $O_2(\cdot)$.
\end{definition}
 \subsection{Quantum Obfuscation}

 \begin{definition}
     Let $\{\mathcal{C}_{\lambda}\}_{\lambda \in \N}$ be a family of circuits and let $\mathcal{X}_\lambda$ be the input space and let $\mathcal{Y}_\lambda$ be the output space of the circuit family.
    A quantum obfuscator consists of two QPT algorithms $(\Obf, \Eval)$ with the following syntax:
    \begin{enumerate}[align=left, leftmargin=2.8em]
        \item[${\Obf(1^\lambda,C)\rightarrow \Tilde{C}}$] The obfuscator takes as input the security parameter $\lambda$ and a classical description of a quantum circuit $C \in \{\mathcal{C}_\lambda\}_{\lambda \in \N}$ and outputs an obfuscation of $C$ which can be classical or quantum.
        \item[${\Eval(\Tilde{C}, \rho_x) \rightarrow \rho_y}$] The evaluation takes as  input the obfuscated program $\Tilde{C}$ and an input $\rho_x \in \mathcal{X}$ and outputs $\rho_y \in \mathcal{Y}$.
    \end{enumerate}
 \end{definition}

\subsubsection{Quantum Indistinguishability Obfuscation}
Several definitions for qiO have come up in the literature. We closely follow the definition of \cite{LC:BroKaz21}. \footnote{The qiO definition from the earlier work of~\cite{AF16} differs in  that they require a weaker notion of functional equivalence for $C_1,C_2$ in item~\ref{qfe:item:sec}. The circuits are required to have a negligible diamond norm but we (following~\cite{LC:BroKaz21}) require a diamond norm of 0. In fact \cite{AF16} show that qiO is impossible to achieve under their definition.}  

\begin{definition} (Quantum Indistinguishability Obfuscation)

    The following three properties are required of a quantum indistinguishability obfuscator:
    \begin{enumerate}[align=left, leftmargin=2.8em]
        
     \item Correctness: The obfuscation scheme is correct if for any circuit $C \in \{\mathcal{C}_{\lambda}\}_{\lambda \in \N}$ there existat a negligible functions $negl(\lambda)$ such that  
     $$||\Eval(\Tilde{C}, \cdot) - C(\cdot)||_\diamond \leq 1-negl(\lambda) $$
     where $\Tilde{C} \leftarrow \Obf(1^\lambda, C)$.
        \item Efficiency: There exists a polynomial $p(\lambda)$ such that for any $C \in \{\mathcal{C}_{\lambda}\}_{\lambda \in \N}$ the size of the obfuscated circuit is only larger by a factor of $p(|C|)$ : $$|\Obf(C)| \leq p(|C|)$$
        \item \label{qfe:item:sec} Security: For any two circuits $C_1, C_2 \in \{\mathcal{C}_{\lambda}\}_{\lambda \in \N}$ that are perfectly equivalent

      $$ ||C_1 - C_2||_\diamond = 0$$
        no QPT distinguisher can distinguish their obfuscation with more than negligible probability:
        $$ \left|\Pr[\Dist(\Obf(C_1)) = 1] - \Pr[ \Dist(\Obf(C_2)) = 1] \right| \leq negl(\lambda)$$
    \end{enumerate}
\end{definition}

\subsubsection{Quantum Virtual Black Box Obfuscation}

\begin{definition}(Quantum Virtual Black Box Obfuscation)
    The following properties are required of a QVBB obfuscator:
    \begin{enumerate}[align=left, leftmargin=2.8em]
        \item Correctness: The obfuscation scheme is correct if for any circuit $C \in \{\mathcal{C}_{\lambda}\}_{\lambda \in \N}$ there existat a negligible functions $negl(\lambda)$ such that  
     $$||\Eval(\Tilde{C}, \cdot) - C(\cdot)||_\diamond \leq 1-negl(\lambda) $$
     where $\Tilde{C} \leftarrow \Obf(1^\lambda, C)$.
        \item Efficiency: There exists a polynomial $p(\lambda)$ such that for any $C \in \{\mathcal{C}_{\lambda}\}_{\lambda \in \N}$ the size of the obfuscated circuit is only larger by a factor of $p(|C|)$ : $$|\Obf(C)| \leq p(|C|)$$
        \item Security: For every QPT adversary $\Adv$, there exists a QPT simulator $\Sim$ with superposition access to its oracle such that for all circuits $C \in \{C_\lambda\}_{\lambda \in \N}$,
        $$ \left|\Pr[\Adv(\Tilde{C}) = 1] - \Pr[\Sim^{C(\cdot)}(1^\lambda, 1^{|C|}) = 1]\right| \leq negl(\lambda)$$
        where $\Tilde{C} \leftarrow \Obf(1^\lambda, C)$.
    \end{enumerate}
\end{definition}
\subsection{Unclonable Encryption}

\begin{definition}(Unclonable Encryption)
\label{qfe:def:ue}
A unclonable encryption scheme consists of three QPT algorithms $(\Keygen, \Enc, \Dec)$
\begin{itemize}[align=left, leftmargin=2.8em]
    \item[$\bm{\Keygen(1^\lambda) \rightarrow (\ek,\dk)}$] $\Keygen$ takes as input the security parameter and outputs an encryption key $\ek$ and a decryption key $\dk$.
    \item[$\bm{\Enc(\ek, \msg) \rightarrow |\ct\rangle}$] $\Enc$ takes as input the encryption key and a message and outputs a quantum ciphertext. 
    \item[$\bm{\Dec(\dk, |\ct\rangle) \rightarrow \msg}$] $\Dec$ takes as input the decryption key and the quantum ciphertext and outputs a message
\end{itemize}
\end{definition}

\begin{definition} (Correctness)
    $$\Pr[\msg = \Dec(\dk, |\ct\rangle) : |\ct\rangle \leftarrow \Enc(\ek,\msg), (\ek,\dk) \leftarrow \Keygen(1^\lambda)] \geq 1 - negl(\lambda)$$
\end{definition}

There are various flavours of unclonable encryption such as secret-key unclonable encryption with quantum decryption keys where $\ek$ is a private classical key and $|\dk\rangle$ is a quantum state (see \cref{qfe:def:ueqdec}) or public-key unclonable encryption where $\ek$ is a classical public key and $\dk$ is a classical decryption key.

\begin{definition} (Unclonable Encryption with Quantum Decryption Keys)
\label{qfe:def:ueqdec}
    An unclonable encryption scheme with quantum decryption keys is defined as in \cref{qfe:def:ue} where $\Keygen$ produces a secret key pair such that the decryption key is a quantum state $|\dk\rangle$ and the encryption key is a classical key. The algorithm $\Keygen$ is pseudodeterministic such that it can produce several copies of the same decryption key.
\end{definition}

\begin{definition} (Unclonable-indistinguishable Security for Secret Key UE)
\label{qfe:def:ueqsec}
    Let $\Adv = (A,B,C)$ be a QPT adversary and let $\lambda \in \N$ be the security parameter.
    \begin{align*}
        &\Exp^{IND-UEQ}_{\Adv}(1^\lambda)\\
        &(\msg_0,\msg_1, st) \leftarrow A(1^\lambda) \text{where} |\msg_0|=|\msg_1|=1\\
        &(\ek,|\dk\rangle^{\otimes 2}) \leftarrow \Keygen(1^\lambda)\\
        & b \leftarrow \bin\\
        & |\ct\rangle \leftarrow \Enc(\ek, \msg_b)\\
        & \rho^{BC} \leftarrow A(|\ct\rangle, st)\\
        & b_B \leftarrow B(\rho^B, |\dk\rangle) \text{ and } b_C \leftarrow C(\rho^C, |\dk\rangle)   \text{ where $B$ and $C$ are not allowed to communicate.}\\ 
        & \text{Output } b_B = b_C = b\\
    \end{align*}
    
     An unclonable encryption scheme is called one-time unclonable-indisintguishable secure if for all $(A,B,C)$ if there exists a negligible function negl such that for all $\lambda \in \N$: 
    $$\Pr[\Exp^{IND-UEQ}_{\Adv}(1^\lambda) = 1] \leq \frac{1}{2} + negl(\lambda)$$
\end{definition}

Such a scheme is presented in \cite{AKY24}, where the authors additionally define a security notion of t-unclonability which allows the adversary to get t copies of the secret key.

\begin{lemma}\cite{AKY24}
    There is a one-time unclonable encryption scheme with quantum decrpytion keys for single bit messages.
\end{lemma}

\begin{definition} (Unclonable-Indistinguishable Security for Public Key Unclonable Encryption with Variable Decryption Keys)
\label{qfe:def:pkeue}

        Let $\Adv = (A,B,C)$ be a QPT adversary and let $\lambda$ be the security parameter.\\
    \begin{align*}
        &\Exp^{UE-VDK}_{\Adv}(1^\lambda)\\
        &(\ek,\dk_0) \leftarrow \Keygen(1^\lambda, r_0), (\ek,\dk_1) \leftarrow \Keygen(1^\lambda, r_1), \text{ where } r_0 = (r,r_0'),\\
        &r_1 = (r,r_1'), r_0',r_1' \leftarrow \bin^{l(\lambda)}, r \leftarrow \bin^{k(\lambda)}\\
        &(\msg_0,\msg_1, \rho_\st) \leftarrow A(1^\lambda, \ek) \text{ where } |\msg_0|=|\msg_1|=n\\
        & b \leftarrow \bin\\
        &|\ct\rangle \leftarrow \Enc(\ek, \msg_b)\\
        &\rho^{BC} \leftarrow A(|\ct\rangle, \rho_\st)\\
        &b_B \leftarrow B(\rho^B, \dk_0) \text{ and } b_C \leftarrow C(\rho^C, \dk_1)  \text{  where B and C are not allowed to}\\ &\text{communicate.} \\
        & \text{Output } b_B = b_C = b
    \end{align*}
     An unclonable encryption scheme is called unclonable-indisintguishable secure if for all $(A,B,C)$ there exists a negligible function negl such that for all $\lambda \in \N$: 
$$\Pr\left[ \Exp^{UE-VDK}_{\Adv}(1^\lambda) = 1\right] \leq \frac{1}{2} + negl(\lambda)$$

\end{definition}

\section{Definition: Quantum Functional Encryption}\label{qfe:sec}

In this section we adapt the definition of Functional Encryption to the Quantum setting. First, we give a defintion for simulation security and then for indistinguishability security. We show that simulation security implies our definition of indistinguishability security. 

\begin{definition}{Quantum Functional Encryption}
\label{qfe:def:qfe}
    Let $\lambda$ be the security parameter and let $(\Setup, $ $\Keygen,$ $ \Enc, \Dec)$ be QPT algorithms.  
\begin{itemize}[align=left]
    \item[$\Setup(1^\lambda) \rightarrow (\mpk, \msk)$] Given the security parameter $\lambda$ output a master public key $\mpk$ and a master secret key $\msk$. 
    \item[$\Keygen(\msk, C) \rightarrow \sk_C$] Given the master secret key and a quantum circuit $C$  output a secret key $\sk_C$.
    \item[$\Enc(\mpk, \rho_\msg) \rightarrow \rho_\ct$] Given the public key $\mpk$ and a message $\rho_\msg$ output a ciphertext $\rho_\ct$.
    \item[$\Dec(\sk_C, \rho_\ct) \rightarrow C(\rho_\msg)$] Given a function secret key $\sk_C$ and ciphertext $\rho_\ct$ which is an encryption of $\rho_\msg$ output the value $C(\rho_\msg)$.
\end{itemize}
\end{definition}

\begin{definition}[Correctness of a functional encryption scheme]
\label{qfe:def:cor}
For all messages $\rho_{\msg z}$, circuits $C$ and random coins used by $\Enc$ and $\Setup$ it holds that 
    \begin{align*}
        &(C(\rho_\msg),\rho_z) = (\Dec(\sk_C,\rho_\ct), \rho_z)
    \end{align*}
    where $\sk_C \leftarrow \Keygen(\msk,C) ,\rho_\ct \leftarrow \Enc(\mpk, \rho_\msg)$and $ (\mpk,\msk)\leftarrow\Setup(\lambda)$
\end{definition}

\subsection{Simulation Based Security Definition}

\begin{definition}[Single-query Sim-Security for QFE]
\label{qfe:def:simsecurity}
Let $\lambda$ be the security parameter and let $\Adv = (\Adv_1, \Adv_2)$ be a QPT adversary and let $\Sim$ be a QPT simulator. 
   \begin{table}[H]
        \centering
        \begin{tabular}{p{6cm}|p{6cm}}
            $\Exp_{\Adv}^{Real}(1^\lambda)$ & $\Exp_{\Adv}^{Ideal}(1^\lambda)$ \\
            $(\mpk,\msk) \leftarrow \Setup(1^\lambda)$&$(\mpk,\msk) \leftarrow \Setup(1^\lambda)$ \\
            $ (\rho_\msg, \rho_\st) \leftarrow \Adv_1^{O_1(\cdot)}(\mpk)$ & $ (\rho_\msg,\rho_\st) \leftarrow \Adv_1^{O_1(\cdot)}(\mpk)$ \\
            $\rho_{\ct} \leftarrow \Enc(\mpk, \rho_\msg)$& $\rho_{ct} \leftarrow \Sim(1^\lambda, \mpk, \mathcal{V})$ \\
            & \quad where $\mathcal{V} = (C, \sk_{C}, C(\rho_\msg), 1^{|\rho_\msg|})$ if $\Adv$\\ & \quad queried $O_1$ on $C$ and $\mathcal{V} = \emptyset$ otherwise.\\
            $\alpha \leftarrow \Adv_2^{O_2(\cdot)}(\rho_{\ct}, \rho_\st)$ & $\alpha \leftarrow \Adv_2^{O_2'(\cdot)}(\rho_\ct, \rho_\st)$\\
            The experiment outputs the state $\alpha$ & The experiment outputs the state $\alpha$\\
        \end{tabular}
    \end{table}
The QFE scheme is single-query simulation-secure if for any adversary $\Adv$ and all messages $\rho_\msg$ there exists a simulator $\Sim$ such that the real and ideal distributions are computationally indistinguishable:
$$ \{ \Exp_{\Adv}^{Real}(1^\lambda)\}_{\lambda \in \N} \approx_c \{ \Exp_{\Adv}^{Ideal}(1^\lambda)\}_{\lambda \in \N} $$.

\textbf{Adaptive vs Non-adaptive security:}
\begin{enumerate}[align=left,leftmargin=2.8em]
    \item Non-adaptive: the adversary $\Adv_1$ is allowed to make one key query to $O_1(\cdot)$ where the oracle $O_1(\cdot)$ is $\Keygen(\msk, C) \rightarrow sk_C$.
    \item Adaptive: the adversary is allowed to make one key query either to $O_1(\cdot)$ or $O_2(\cdot)$ ($O'_2(\cdot)$ in the ideal world) where $O_1(\cdot)$ and $O_2(\cdot)$ are $\Keygen(\msk, C) \rightarrow sk_C$ and $O_2'(\cdot)$ is a $\Keygen$ oracle controlled by the simulator $\sk_C \leftarrow \Sim(1^\lambda, \msk, C, C(\rho_\msg), 1^{|\rho_\msg|})$. The simulator is stateful, in this invocation $\Sim$ has access to the state of the simulator from it's first invocation where it produced the ciphertext. 
\end{enumerate}

\end{definition}

\subsection{Indistinguishability Based Security Definition}
\label{qfe:sec:inddef}
In this section we comment on potential issues when trying to find an appropriate indistinguishability-based definintion of functional encryption for the quantum setting.
The simulation-based definition is generally preferred as indistinguishability-based security does not capture a meaningful security notion for some functions~\cite{oneil10, TCC:BonSahWat11}. Nevertheless indistinguishability-based security can be easier to achieve and still has many important applications as we can see in the extension to the multi-input setting in \cref{qfe:sec:qmife}.

First we give the definition for the IND-security experiment and then we discuss the notion of admissible queries in depth. 

\begin{definition}[Single-Query IND-Security for QFE]
\label{qfe:def:indsecurity}

Let $\lambda$ be the security parameter and let $\Adv = (\adv_0,\Adv_1)$ be a QPT adversary.
   \begin{align*}
            &\Exp^{IND}_{\Adv,b}(1^\lambda)\\
            &(\mpk,\msk) \leftarrow \Setup(1^\lambda)\\
             &(\rho_{\msg_0}, \rho_{\msg_1}, \rho_\st) \leftarrow \Adv_0^{\sk_C \leftarrow \Keygen(\msk, \cdot )}(\mpk), \text{where $\rho_{\msg_0}$ and $\rho_{\msg_1}$ are admissible queries} \\ &\text{for the circuit $C$ that $\Adv$ queries. }\\ 
            &\rho_\ct \leftarrow \Enc(\mpk, \rho_{\msg_b})\\
            & b' \leftarrow \Adv_1^{O(\cdot)}(\mpk, \rho_\ct, \rho_\st)
    \end{align*}
The FE scheme is called secure if for any adversary $\Adv$ that makes admissible queries (\cref{qfe:def:admis}) it holds that
$$\left|\Pr\left[1 \leftarrow \Exp_{\Adv, b=0}^{Ind}\right] - \Pr\left[1 \leftarrow \Exp_{\Adv, b=1}^{Ind}\right]\right| \leq negl(\lambda)$$
where the random coins are taken over the randomnes of $\Adv$, $\Setup, \Keygen$ and $\Enc$.

\textbf{Adaptive vs. Non-adaptive security}
\begin{itemize}
    \item The scheme is called non-adaptively secure if the the adversary only queries the $\Keygen$ oracle before receiving a ciphertext. Then the oracle $O(\cdot)$ is the empty oracle. 
    \item The scheme is called adaptively secure if the adversary can either query the $\Keygen$ oracle before or after receiving the ciphertext. Then the oracle $O(\cdot)$ is the function $\Keygen(\msk, \cdot)$.
\end{itemize}
\end{definition}

In the classical setting admissible queries are defined as $C({\msg_0}) = C({\msg_1})$. To adjust this definition to the quantum setting we have to redefine the condition that the quantum circuit has the same output on the inputs $\rho_{\msg_0}$ and $\rho_{\msg_1}$. A natural first attempt to define admissible queries $\rho_{\msg_0}, \rho_{\msg_1}$ is to use the trace distance of the output states since the trace distance bounds the adversaries probability of distinguishing two quantum states

$$ \T\left(C(\rho_{\msg_0}), C(\rho_{\msg_1})\right) \leq negl(\lambda)$$ 

This definition is not sufficient as can be seen in the following scenario: $\Adv$ creates the states $\rho = |EPR\rangle\langle EPR|$ and $\sigma = |EPR \rangle \langle EPR|$ and gives one qubit each to the experiment $\rho_{\msg_0} = \rho_1$ and $ \rho_{\msg_1} = \sigma_1$. $\Adv$ queries the identity circuit and receives $\rho_\ct$.  The states $\rho_{\msg_0}$ and $\rho_{\msg_1}$ have trace distance 0 since they are both the maximally mixed state. $\Adv$ can decrypt $\rho_\ct$ using the function secret key and check with non-negl probability which quibt it is by applying a coherent measurement on the qubit remaining in it's internal state and the received qubit. 

The above attack is not applicable in the simulation-based setting. The scheme that we proved secure under Sim-security allows an adversary to stay entangled with the challenge message. This entanglement is maintained by the encryption procedure or the simulator respectively. 

An alternative approach to defining IND-security would be to take the adversaries internal state into account. The messages $\rho_{\msg_0} = \sum_i p_i \rho_{\msg_{0,i}} , \rho_{\msg_1} = \sum_i q_i \rho_{\msg_{1,i}}$ are admissible queries if 
\begin{align}
\label{qfe:admis1}
    \T\left(\sum_{i} p_i C(\rho_{\msg_{0,i}}) \otimes \rho_{A_i}, \sum_{i} q_i C(\rho_{\msg_{1,i}}) \otimes \rho_{A_i}\right) \leq negl(\lambda).
\end{align}

where $\rho_{A}$ is the adversary's internal state.

For many functionalities this would enforce the adversary to stay unentangled with the challenge message queries. The definition might still be useful in some applications, as for example messages that are not chosen by the adversary fall into this category.

To allow the adversary more freedom and in particular to enable the adversary to stay entangled with a part of the challenge message we can allow the following way of querying messages. 

\begin{definition}(Admissible queries)
\label{qfe:def:admis}
    For a challenge message $\rho^{EU}_{\msg_{b}}$ the adversary specifies a register $E$ that is encrypted and a register $U$ that will be returned unencrypted to the adversary. The message $\rho^{EU}_{\msg_{1-b}}$ which is not used as the challenge is not returned to the adversary. Then the challenge queries have to fulfill: 

\begin{align}
\label{qfe:admis2}
     \T\left(\sum_{i} p_i C(\rho^E_{\msg_{0,i}}) \otimes \rho^U_{\msg_{0,i}} \otimes \rho_{A_i}, \sum_{i} q_i C(\rho^E_{\msg_{1,i}}) \otimes \rho^U_{\msg_{1,i}} \otimes \rho_{A_i}\right) \leq negl(\lambda)
\end{align}

where $\rho_{A}$ is the adversary's internal state and $C$ is the circuit that the adversary queries.
\end{definition}

In practice this allows for any entanglement to be moved into the challenge query such that the state of the adversary is unentangled with the message queries and the state can be written as $\rho_{\msg^{EU}_0} \otimes \rho_{\msg^{EU}_1} \otimes \rho_\Adv$. This simplifies the check if the query is admissible to 
\begin{align*}
    &\T\left(\sum_{i} p_i C(\rho^E_{\msg_{0,i}}) \otimes \rho^U_{\msg_{0,i}} \otimes \rho_{A}, \sum_{i} q_i C(\rho^E_{\msg_{1,i}}) \otimes \rho^U_{\msg_{1,i}} \otimes \rho_A \right) \\
    & =\T\left(\sum_{i} p_i C(\rho^E_{\msg_{0,i}}) \otimes \rho^U_{\msg_{0,i}} , \sum_{i} q_i C(\rho^E_{\msg_{1,i}}) \otimes \rho^U_{\msg_{1,i}}\right) \leq negl(\lambda).\\
\end{align*}

 Useful special cases of \cref{qfe:def:admis} are
\begin{enumerate}[align = left ,leftmargin= 2.8em]
    \item Classical messages. For classical messages the definition reduces to the classical definition of admissibility since 
    \begin{align*}
        \T\left(C({\msg_{0}}) \otimes \rho_{A}, C({\msg_{1}}) \otimes \rho_{A}\right) = 0.
    \end{align*}
    exactly when $C(\msg_0) = C(\msg_1)$.
    \item Defining both $\rho_{\msg_0}$ and $\rho_{\msg_1}$ with respect to a single quantum state $\sigma$. In the IND-security game the adversary might hold a single copy of a special quantum state $\sigma$ which he would like to use for defining both messages $\rho_{\msg_0}$ and $\rho_{\msg_1}$. Since the experiment only creates a single ciphertext and discards the other message we can allow the adversary to only provide a single copy of $\sigma$ and define $\rho_{\msg_0}$ and $\rho_{\msg_1}$ to each contain the state $\sigma$. 
\end{enumerate}
This definition of IND-security is implied by simulation secure quantum functional encryption. 

\begin{lemma}
\label{qfe:lem:simtoind}
    A QFE scheme that is single-query (non)-adaptively  SIM-secure (\cref{qfe:def:simsecurity}) is also single-query (non)-adaptively IND-secure (\cref{qfe:def:indsecurity}).
\end{lemma}
The proof can be found in \cref{qfe:app:lem1proof}.

\section{Construction: Quantum Functional Encryption}

We construct a single-query adaptively secure functional encryption scheme for quantum circuits. We start by constructing a simple functional encryption scheme for a single circuit where the circuit has to be fixed ahead of time. Then we use this construction to achieve single-query adaptively secure functional encryption for polynomial sized circuits. Our construction follows the ideas used by \cite{CCS:SahSey10, C:GorVaiWee12} in the classical setting. They show how to leverage classical randomized encodings to achieve classical functional encryption. Similarly, quantum randomized encodings can be used to achieve quantum functional encryption.

\subsection{QFE for a Single Circuit}
\label{qfe:sec:unitaryfe}
First we construct a quantum functional encryption scheme that only allows to evaluate a circuit family consisting of one circuit $\mathcal{C} = \{C_\lambda\}_{\lambda in \N}$ with fixed input size $n = poly(\lambda)$ and output size $d = poly(\lambda)$. To achieve this construction we make use of the Quantum One Time Pad and a classical FE scheme that allows functional encryption for the identity circuit $\IdFE = (\IdFE.\Setup, \IdFE.\Keygen, $ $\IdFE.\Enc, \IdFE.\Dec)$. Such a scheme is constructed in~\cite{C:GorVaiWee12}. Let $D(\Hilb_{n})$ be the input space, let $D(\Hilb_{d})$ be the output space and let the circuit be denoted as $C$.

\begin{itemize}[align=left,leftmargin=2.8em]
    \item[$\bm{\Setup(1^\lambda) \rightarrow (\mpk,\msk)}$] Run the classical $\IdFE$ scheme to obtain the keys $(\pk,\sk) \leftarrow \IdFE.\Setup(1^\lambda)$. Output $(\mpk=\pk, \msk=\sk)$.
    \item [$\bm{\Enc(\mpk, \rho_\msg) \rightarrow \ct}$] Sample a pair of keys for the QOTP $(a,b)$, where $a, b \in \bin^d$. Compute $$\rho_{\ct_0} = \Xgate^{a} \Zgate^{b} C(\rho_\msg) $$ Encrypt the QOTP keys using the classical FE scheme $$\ct_1 = \IdFE.\Enc(\mpk, (a,b))$$ Output $\ct = (\rho_{\ct_0}, \ct_1)$. 
    \item[$\bm{\Keygen(\msk) \rightarrow \sk^*}$] Run the $\IdFE$ scheme to obtain the  secret key $\sk^* = \IdFE.\Keygen(\msk)$.
    \item[$\bm{\Dec(\sk^{*}, \ct) \rightarrow \rho_\msg}$] Given $\ct = (\rho_{\ct_0}, \ct_1)$ use the key $\sk^*$ to obtain the QOPT keys $(a,b) = \IdFE.\Dec(\sk^*, \ct_1)$ and then decrypt the quantum state $$ \rho_\msg = \Xgate^{a} \Zgate^{b} \rho_{\ct_0}$$
\end{itemize}

\begin{theorem}
    Given a classical FE scheme for the identity circuit that fulfills adaptive sim-security, there exists an adaptively sim-secure QFE scheme for a single circuit. 
\end{theorem}

\begin{proof}
~\paragraph{Correctness} Due to the correctness of the classical FE scheme and the correctness of the QOTP the scheme is correct. 

~\paragraph{Security} We define a simulator $\Sim$ for the scheme. The adversary can either query the key first and then obtain the ciphertext or obtain the ciphertext first and then the key. We distinguish the simulator's behaviour in these two cases. 
\begin{enumerate}[align=left,leftmargin=2.8em]
    \item The adversary queries non-adaptively, i.e. it queries the key first. That means the simulator obtains $C(\rho_\msg)$. The simulator creates the ciphertext as the honest encryption algorithm would. 
    \item The adversary queries adaptively, i.e. it queries the ciphertext first. The simulator needs to create a ciphertext without knowledge of the value it should later decrypt to. The simulator creates d EPR pairs and sends one qubit of each EPR pair to the adversary and keeps the other qubit of each EPR pair. The classical ciphertext is simulated via the simulator of the classical FE scheme: 
    $$ \ct = \IdFE.\Sim(\mpk, |x|=2d) $$
    
    When the adversary queries the key, the simulator learns $C(\rho_\msg)$ and performs the teleportation circuit using $C(\rho_\msg)$ and the halves of the EPR pairs which he holds. $\Sim$ obtains the correction keys $(a,b) \in \bin^{d}$ and creates the key using the simulator for the classical $\IdFE$-scheme: 
    $$  sk^* = \IdFE.\Sim(\sk, (a,b))$$.
    The simulator outputs $\sk^*$.
\end{enumerate}
In the case of a non-adaptive query the simulator behaves as the experiment in the real world. Therefore real and ideal experiments are indistinguishable. For the case of an adaptive query we establish security via a series of hybrids:

\paragraph{Hybrid 0:} This is the real world, where the ciphertext is created by the Encryption algorithm 

\paragraph{Hybrid 1:} In this Hybrid we use the simulator of the classical $\IdFE$-scheme to simulate the ciphertext in case of an adaptive query. The quantum state part of the ciphertext is  created honestly and the corresponding encryption keys are used to answer the key query using the simulator of the $\IdFE$-scheme.  

\begin{claim}
    Hybrid 0 and Hybrid 1 are computationally indistinguishable. 
\end{claim}

\begin{proof}
    Due to the adaptive security of the $\IdFE$-scheme this change is not noticeable to the adversary. An adversary that can distinguish between Hybrid 0 and Hybrid 1 could distinguish between the real and simulated experiment of the $\IdFE$ scheme. 
\end{proof}

\paragraph{Hybrid 2:}This is the Ideal world where the simulator $\Sim$ runs as defined above.

\begin{claim}
    Hybrid 1 and Hybrid 2 are perfectly indistinguishable. 
\end{claim}

\begin{proof}
    The simulator creates $d$ EPR pairs and sends one qubit each as a ciphertext $\rho_{\ct_0}$. Upon receiving $(\rho_{\ct_0},\ct_1)$ the adversary cannot distinguish $\rho_{\ct_0}$ in Hybrid 1 from the state in Hybrid 2 since 1 qubit of an EPR pair appears as a maximally mixed state, the same as a state encrypted under the QOTP. Since $\ct_1$ is a ciphertext simulated by $\IdFE$ as in the previous Hybrid it contains no information about the QOTP keys. 
    Upon receiving the key query the simulator obtains $C(\rho_\msg)$ and teleports the state through the corresponding EPR pairs and obtains the correction keys $(a,b)$. The teleported state the adversary now holds is $X^a Z^b C(\rho_\msg)$ which is equivalent to a QOTP encrypted state with the key $(a,b)$. The keys are revealed using the $\IdFE$ simulator. 
\end{proof}

\end{proof}

\subsection{QFE for a poly-sized family of circuits}
\label{qfe:sec:polyqfe}

In this section we construct a QFE scheme for circuits of polynomial size.  We need the following building blocks: 

Let $\OneFE = (\Setup, \Keygen, \Enc, \Dec)$ be the single circuit QFE scheme from the previous section. Let $\TwoFE = (\Setup, \Keygen, \Enc, \Dec)$ be a classical FE scheme for a family of two circuits~\cite{C:GorVaiWee12}. Let $\QRE = (\Encode, \Decode)$ be a quantum randomized encoding scheme that is also decomposable. In particular we will use the quantum garbled circuits construction of \cite{BY20} which has the special property that if there is a classical part of the input the encoding procedure is classical. 

Let $U$ be a universal quantum circuit, that is on inputs $\rho_\msg$ and $C$ it evaluates to  $U(C,\rho_\msg) = C(\rho_\msg)$. Let the decription of $C$ have length $l$ and $\rho_\msg$ be a quantum state of dimension $n$. Then we can create a randomized encoding of $U(C,\rho_\msg)$ where due to the decomposability the randomized encoding can be created in independent pieces where each piece only depends on one bit of the input. Let $R$ denote classical randomness and $e$ denote a set of EPR pairs, then the randomized encoding $\Tilde{U}$ can be written as 
\begin{align*}
    \Encode(U,C,\rho_\msg,R,e) &= \Tilde{U}(C,\rho_\msg,R,e)\\
    &= (\Tilde{U}_1(C[1],R,e_1), \dots, \Tilde{U}_l(C[l],R,e_1), \Tilde{U}_x(\rho_\msg,R,e_2), \Tilde{U}_{off}(R,e_3))
\end{align*}

where $e_1,e_2,e_3$ are disjoint subsets of the qubits contained in the set of EPR pairs e. 

To construct FE for a poly sized family of circuits we use $l$ instances of the classical $\TwoFE$ scheme for the family of two circuits $\{f_{C[i] = 0}, f_{C[i]=1}\}$:
\begin{align*}
    f_{C[i] = 0}(R,t) = \Tilde{U}_i(0,R,t)\\
    f_{C[i] = 1}(R,t) = \Tilde{U}_i(1,R,t)
\end{align*}
where $t$ is a classical bit that can be obtained from measuring the EPR pairs from the set $e_1$. 

During encryption we create $l$ classical ciphertexts that can later be opened to the description of the circuit using the $\Keygen$ algorithm of $\TwoFE$. Given a description of $C$ the $i'th$ ciphertext is opened such that it decrypts to $\Tilde{U}_i(C[i],R,t)$, the randomized encoding of the $i'th$ bit of the description of $C$. 

Additionally we use two instances of  the quantum $\OneFE$ scheme for the circuits:
$$ f_{in}(\rho_\msg,e,R) = \Tilde{U}_{in}(\rho_\msg,R,e)$$
$$ f_{off}(e,R) = \Tilde{U}_{off}(R,e)$$

Putting everything together we can see that our encryption procedure produces the individual pieces of the decomposable randomized encoding scheme by relying on simpler functional encryption primitives. The final output $C(\rho_\msg)$ can be obtained by decrypting the individual parts of the ciphertext and then the result can be decoded. 

Let $\lambda \in \N$ be the security parameter and let $\mathcal{M} = D(\Hilb_s)$ where $s = poly(\lambda)$ be the message space. Let $\mathcal{C} = \{C_\lambda\}_\lambda$ be a family of quantum circuits with inputs of size $s$, outputs of size $t = poly(\lambda)$ and classical description of size $l = poly(\lambda)$.

\begin{itemize}[align=left,leftmargin=2.8em]
    \item[$\bm{\Setup(1^\lambda) \rightarrow (\mpk,\msk)}$] Create $l$ keys for the $\TwoFE$-scheme: $$(\pk_i, \sk_i) \leftarrow \TwoFE.\Keygen(1^\lambda) \quad for~ i \in 1, \ldots, l$$
    where the i-th keypair is associated with the circuit family $\{f_{C[i]=0}, f_{C[i]=1}\}$.
    
    Run the $\OneFE$ scheme twice, once for the circuit $f_{in}$ and once for the circuit $f_{off}$.
    \begin{align*}
    (\pk_{in}, \sk_{in}) \leftarrow \OneFE.\Keygen(1^\lambda)\\
    (\pk_{off}, \sk_{off}) \leftarrow \OneFE.\Keygen(1^\lambda)
    \end{align*}
    Output  $(\mpk = (\pk_1, \ldots, \pk_l, \pk_{in}, \pk_{off}), \msk = (\sk_1, \ldots, \sk_l, \sk_{in}, \sk_{off}))$.
    \item [$\bm{\Enc(\mpk, \rho_\msg \in \mathcal{M}) \rightarrow \rho_\ct}$] Sample $R \leftarrow \mathcal{R}$ and sample $l+n+k$ EPR pairs. The EPR pairs are split into 3 groups $E^l = \{(e^l_{i,1},e^l_{i,2})\}_{i \in [l]}$,  $E^n = \{(e^n_{i,1},e^n_{i,2})\}_{i \in [n]}$ and $E^k = \{(e^k_{i,1},e^k_{i,2})\}_{i\in [k]}$.
    
    For $i \in 1, \ldots, l$ take the first qubit of each EPR pair in the group $E^l$ and measure it in the computational basis to obtain $t_i$, then compute
    $$ ct_i \leftarrow \TwoFE.\Enc(\pk_i,R,t_i)$$
    Use the quantum $\OneFE$ scheme to compute the ciphertext
    $$\rho_{ct_{in}} \leftarrow \OneFE.\Enc(\pk_{in}, \rho_\msg, R, \{e^n_{i,1}\}_{i \in [n]})$$
    and compute the ciphertext
    $$\rho_{ct_{off}} \leftarrow \OneFE.\Enc(\pk_{off}, R, \{e^n_{i,2}\}_{i \in [n]}, \{e^l_{i,2}\}_{i \in [l]}, E^k )$$
    
    Output $\rho_\ct = (\{\ct_i\}_{i \in [l]},\rho_{\ct_{in}}, \rho_{\ct_off})$.
    \item[$\bm{\Keygen(\msk, C \in \mathcal{C}) \rightarrow \sk_C^*}$] For $i \in 1,\ldots,l$ create
    $$ sk^*_i = \TwoFE.\Keygen(\sk_i, f_{C[i]})$$
    and create
    $$ sk_{in}^* \leftarrow \OneFE.\Keygen(\sk_{in},f_{in})$$
    $$ sk_{off}^* \leftarrow \OneFE.\Keygen(\sk_{off}, f_{off})$$
    Output $sk_C^* = (\sk^*_1, \ldots, \sk^*_l, \sk_{in}^*, \sk_{off}^*)$.
    \item[$\bm{\Dec(\sk_C^{*}, \rho_\ct) \rightarrow \rho_\msg}$] For $i \in 1,\cdots,l$ decrypt
    $$ \Tilde{U}(C[i],R,t_i) = \TwoFE.\Dec(\sk^*_i, ct_i)$$
    and create
    $$ \Tilde{U}(\rho_\msg,R,e) \leftarrow \OneFE.\Dec(\sk^*_{in},\rho_{ct_{in}})$$
    $$ \Tilde{U}(R,e) \leftarrow \OneFE.\Dec(\sk^*_{off}, \rho_{ct_{off}})$$
    Output $y = \Decode(\Tilde{U}(C[1],R,t_1), \ldots, \Tilde{U}(C[l],R,t_l),\Tilde{U}(\rho_\msg,R,e), \Tilde{U}(R,e))$.
\end{itemize}

\begin{theorem}
    Given an adaptively sim-secure classical FE scheme for a family of two circuits, an adaptively sim-secure QFE scheme for a single circuit and a QGC scheme, there exists an adaptively sim-secure QFE scheme for poly sized circuits.
\end{theorem}
\begin{proof}

~\paragraph{Correctness}
\begin{align*}
    &\Dec(\sk_C,\Enc(\msk, \rho_\msg)) = \Dec(\sk^*_C, ct_1, \ldots ct_l, \rho_{ct_{in}}, \rho_{ct_{off}})\\
    &= \Decode(\TwoFE.\Dec(\sk^*_1, ct_1), \ldots, \TwoFE.\Dec(\sk^*_l, ct_l), \OneFE.\Dec(\sk^*_{in}, \rho_{ct_{in}}), \\
    & \quad \OneFE.\Dec(\sk^*_{off}, \rho_{ct_{off}}))\\
    &= \Decode(\Tilde{U}(C[1],R,t_1), \ldots, \Tilde{U}(C[l],R,t_l),\Tilde{U}(\rho_\msg,R,e), \Tilde{U}(R,e))\\
    &= \Decode(\Encode(U,C,\rho_\msg,R,e))\\
    &= C(\rho_\msg)
\end{align*}

~\paragraph{Security}
We define a simulator $\Sim$ for the scheme. We distinguish whether the adversary makes an adaptive or non-adaptive query. 
\begin{enumerate}[align=left,leftmargin=2.8em]
    \item The adversary queries non-adaptively. The simulator obtains $C,C(\rho_\msg)$ and creates the ciphertext as follows:
    \begin{enumerate}[align=left,leftmargin=2.8em]
        \item Create the randomized encoding using the simulator of the QRE scheme. 
        $$(\Hat{U}(C[1],R,t_1), \ldots, \Hat{U}(C[l],R,t_l),\Hat{U}(\rho_\msg,R,e), \Hat{U}(R,e)) \leftarrow \QRE.\Sim(C(\rho_\msg), \mathcal{T}_C)$$
        where $\mathcal{T}_C$ is the topology of the circuit $U(\cdot)$.
        \item Create the ciphertexts using the simulator of the classical $\TwoFE$ scheme and the quantum $\OneFE$ scheme for the non-adaptive case to create ciphertexts:
        \begin{align*}
            ct_i &\leftarrow \TwoFE.\Sim(\pk_i, \Hat{U}(C[i],R,t_1)) \quad for ~i \in [l]\\
            \rho_{ct_{in}} &\leftarrow \OneFE.\Sim(\pk_{in}, \Hat{U}(\rho_\msg,R,e))\\
            \rho_{ct_{off}} &\leftarrow \OneFE.\Sim(\pk_{off}, \Hat{U}(R,e))
        \end{align*}
    \end{enumerate}
    \item The adversary queries adaptively. The simulator has to create a ciphertext without knowing the evaluation result. 
    \begin{enumerate}[align=left,leftmargin=2.8em]
        \item Use the simulator of the classical $\TwoFE$ scheme and the quantum $\OneFE$ scheme for the adaptive case to create ciphertexts:
        \begin{align*}
            (ct_i, st_i) &\leftarrow \TwoFE.\Sim(\pk_i, 1^{|C[i]|+|R|+|t_i|}) \quad for ~i \in [l]\\
            (\rho_{ct_{in}}, st_{in}) &\leftarrow \OneFE.\Sim(\pk_{in}, 1^{|\msg|+|R|+|e|})\\
            (\rho_{ct_{off}}, st_{off}) &\leftarrow \OneFE.\Sim(\pk_{off}, 1^{|R|+|e|})
        \end{align*}
        \item Upon receiving a key query $\Sim$ obtains $C(\rho_\msg)$ and can create the randomized encoding. 
        $$(\Hat{U}(C[1],R,t_1), \cdots, \Hat{U}(C[l],R,t_l),\Hat{U}(\rho_\msg,R,e), \Hat{U}(R,e)) \leftarrow \QRE.\Sim(C(\rho_\msg))$$
        \item Then $\Sim$ creates the key by using the simulator of the underlying FE schemes.
        \begin{align*}
            \sk^*_i &\leftarrow \TwoFE.\Sim(\sk_i, \Hat{U}(C[i],st_i, R,t_i)) \quad for ~i \in [l]\\
            \sk^*_{in} &\leftarrow \OneFE.\Sim(\sk_{in}, st_{in}, \Hat{U}(\rho_\msg,R,e))\\
            \sk^*_{off} &\leftarrow \OneFE.\Sim(\sk_{off}, st_{off}, \Hat{U}(R,e))
        \end{align*}
    \end{enumerate}
\end{enumerate}

\paragraph{Hybrid 0} This is the real world.

\paragraph{Hybrid i} For $i \in \{1, \dots, l\}$. Sample $R$ and $E^l, E^n, E^k$ as in $\Enc$.

For $1\leq j<l$ let the ciphertexts be created honestly: 
$$ ct_j \leftarrow \TwoFE.\Enc(pk_j,R,t_j)$$

In the non-adaptive case: 

For $i \leq j \leq l$ create the partial randomized encoding and simulate the ciphertext using the simulator of the underlying scheme: 
$$ ct_j \leftarrow \TwoFE.\Sim(\Tilde{U}(C[i],R,t_i))$$

In the adaptive case:

For $i \leq j \leq l$ simulate the ciphertext using the simulator of the underlying scheme: 
$$ ct_j \leftarrow \TwoFE.\Sim(\pk_i, 1^{|C[i]|+|R|+|e_i|})$$

Create $\rho_{ct_{in}}, \rho_{ct_{off}}$ honestly. 

\begin{claim}
    Hybrids 0 to l  are indistinguishable up to negligible probability.
\end{claim}
\begin{proof}
    To show indistinguishability of each pair of games we can invoke the security of the $\TwoFE$ scheme. A distinguisher between the Hybrids can break the security of the $\TwoFE$ scheme. 
   
\end{proof}

\paragraph{Hybrid l+1, Hybrid l+2} 

For ciphertexts $\rho_{ct_{in}}$ and $\rho_{ct_{off}}$ use the simulator to create the ciphertexts in Hybrid l+1 and Hybrid l+2 respectively. 

In the non-adaptive case:

$$ \rho_{ct_{in}} \leftarrow \OneFE.\Sim(\pk_{in}, \Tilde{U}(\rho_\msg,R,E^n))$$
$$ \rho_{ct_{off}} \leftarrow \OneFE.\Sim(\pk_{off}, \Tilde{U}(R,E^k))$$

In the adaptive case: 

$$ \rho_{ct_{in}} \leftarrow \OneFE.\Sim(\pk_{in}, 1^{|\msg|+|R|+|E^n|})$$
$$ \rho_{ct_{off}} \leftarrow \OneFE.\Sim(\pk_{off}, 1^{|R|+|E^k|})$$

\begin{claim}
    Hybrids l and l+1 are indistinguishable as well as Hybrids l+1 and l+2 up to negligible probability.
\end{claim}
\begin{proof}
    To show indistinguishability of each pair of games we can invoke the security of the $\OneFE$ scheme. A distinguisher between the Hybrids can break the security of the $\OneFE$ scheme. 
\end{proof}

\paragraph{Hybrid l+3} 
In the non-adaptive case:
Upon receiving the key query $C, C(\rho_\msg)$ use the simulator of the randomized encoding to create $\Hat{U}(C(\rho_\msg),\mathcal{T}_C) \leftarrow \Sim$ and use the simulated randomized encoding to create the ciphertext instead of the real randomized encoding. 

In the adaptive case: 
Upon receiving the key query $C, C(\rho_\msg)$ use the simulator of the randomized encoding to create $\Hat{U}(C(\rho_\msg),\mathcal{T}_C) \leftarrow \Sim$ and use the simulated randomized encoding to answer the secret key query. 

This is the ideal world.

\begin{claim}
    Hybrids l+2 and l+3 are indistinguishable up to negligible probability.
\end{claim}
\begin{proof}
    Due to the indistinguishability of the simulated randomized encoding from the real randomized encoding the Hybrids are indistinguishable. 
\end{proof}

\end{proof}
\section{Unclonable Functional Encryption}

In this section we define and construct an unclonable functional encryption scheme.  Security requires that two participants who try to copy a ciphertext can obtain independently generated function secret keys for any circuit and cannot both guess which messages out of two challenge messages was encrypted. When the function secret keys are fixed to be the identity circuit this implies a public-key unclonable encryption scheme with variable decryption keys (\cref{qfe:def:pkeue}).

\subsection{Definition}
An unclonable functional encryption scheme is defined by the syntax and correctness properties of a $\QFE$ scheme, see \cref{qfe:def:qfe} and \cref{qfe:def:cor}.

\begin{definition}(Non-adaptive Unclonable-Indistinguishable Functional Encryption)
\label{qfe:def:UEQFE}

Let $\lambda$ be the security parameter, let $\Adv = (A,B,C)$ be a QPT adversary and let $\mathcal{C}_\lambda$ be a family of circuits.
   \begin{align*}
            &\Exp^{QFE-UE-IND}_{\Adv}(1^\lambda)\\
            &(\mpk,\msk) \leftarrow \Setup(1^\lambda)\\
             &(\rho_{\msg_0}, \rho_{\msg_1}, \rho_{\st}, C_B, C_C) \leftarrow A(1^\lambda, \mpk)\\
             & b \leftarrow \bin
            &\rho_\ct \leftarrow \Enc(\mpk, \rho_{\msg_b})\\
            & \sk_{C_B} \leftarrow \Keygen(\msk, C_B), \sk_{C_C} \leftarrow \Keygen(\msk, C_C)\\
            & \rho_{BC} \leftarrow A(\rho_\ct, \rho_{\st})\\
            & b_B \leftarrow B(\mpk, \rho_\ct, \rho_{\st_B}, \sk_{C_B})\\
            & b_C \leftarrow C(\mpk, \rho_\ct, \rho_{\st_C}, \sk_{C_C})\\
            & Output b_B = b_C = b
    \end{align*}
The FE scheme is called unclonable-indistinguishably secure if for any adversary $\Adv = (A,B,C)$ and  any $C_B, C_C \in \mathcal{C}_\lambda$
$$\Pr\left[ \Exp^{QFE-UE-IND}_{\Adv}(1^\lambda) = 1 \right] \leq \frac{1}{2} + negl(\lambda)$$
where the random coins are taken over the randomness of $\Adv$, $\Setup, \Keygen$ and $\Enc$.
\end{definition}

For brevity we often refer to this definition as unclonable functional encryption. 

\begin{remark}
   An adaptive security notion of unclonable functional encryption can be defined by giving each $B$ and $C$ oracle access to the $\Keygen$ functionality instead of $\Adv$ outputting a description of the circuits for which secret keys should be produced.  
\end{remark}

\subsection{Construction}
\label{sec:UFEconstruction}
We need the following components:

\begin{itemize}
    \item Let $\QFE = (\Setup, \Keygen, \Enc, \Dec)$ be a non-adaptive IND-secure QFE scheme.
    \item Let $\UEQ = (\Keygen, \Enc, \Dec)$ be a one-time unclonable-indistinguishable encryption scheme for single bit messages with quantum decryption keys of size $l(\lambda)$ and ciphertext size $t(\lambda)$ and let $s(\lambda)$ be an upper bound on the description of the decrpytion circuit.
\end{itemize}

We highlight that the construction is solely based on the $\QFE$ scheme and the $\UEQ$ scheme is not explicitly used, we only require the existence of a $\UEQ$ scheme. The unclonable $\QFE$ scheme is universal which means that the unclonable security property holds as long as any $\UEQ$ scheme with the mentioned properties is secure, such as ~\cite{AKY24}. 

The construction relies on a QFE scheme for the family of circuits $\mathcal{U}_\lambda$ = $\{U_{p(\lambda),l(\lambda), s(\lambda), n(\lambda)}\}_{\lambda \in \N}$ which has the following structure:

\begin{align*}
        U_{(C,a,b)}&(\rho_{\msg_0},\rho_{\msg_1},|\dk_0\rangle,|\dk_1\rangle,\rho_{UE}, C_{Dec}, f) =\\
        &\text{ if } f = 0 \text{ output } C(\rho_{\msg_0}) \\
        &  \text{ if } f = 1 \text{ do: }\\
        & \quad \text{Compute }|\dk_0'\rangle = X^{a} Z^{b} |\dk_0\rangle \text{ and } |\dk_1'\rangle = X^{a} Z^{b} |\dk_1\rangle\\
        & \quad \text{Measure the first $\lambda$ bits of } |\dk_0'\rangle \text{ in the computational basis,}\\
        & \quad \text{if they are all 0 remove them and set } |\dk^*\rangle = |\dk_0'\rangle\\
        & \quad \quad \text{else measure the first $\lambda$ bits of } |\dk_1'\rangle \text{ in the computational basis,}\\
        & \quad \quad \text{if they are all 0 remove them and set } |\dk^*\rangle = |\dk_1'\rangle\\
        & \quad \quad \text{else if both checks fail output }\bot\\
        & \quad \text{Interpret $C_{Dec}$ as the description of a circuit for decryption:  } U(C_{Dec}, |\dk^*\rangle, \rho_{UE}) = b\\
        & \quad \text{Output } C(\rho_{\msg_b})
    \end{align*}

Then the following is an unclonable functional encryption scheme for a family of quantum circuits $\mathcal{C} = \{C_{\lambda}\}_{\lambda \in \N}$ with classical description of size $p(\lambda)$ inputs in $\mathcal{X} = \mathcal{D}(\Hilb_n)$.

\begin{itemize}[align=left,leftmargin=2.8em]
    \item[$\bm{\Setup(1^\lambda, r) \rightarrow (\mpk,\msk)}$] Run $(\mpk,\msk) = \QFE.\Setup(1^\lambda,r)$.
    
    Output $(\mpk,\msk)$.
    
    \item[$\bm{\Keygen(1^\lambda, C \in \mathcal{C_\lambda}, r') \rightarrow \sk_C}$]
    
     Sample random strings $a,b \leftarrow \bin^{l(\lambda)+s(\lambda)}$  using randomness $r'$. 
    
    Run $\sk_C \leftarrow \QFE.\Keygen(\msk, U_{(C,a,b)})$.

    Output $\sk_C$.
    
    \item[$\bm{\Enc(\mpk, \rho_\msg \in \mathcal{X}) \rightarrow \rho_\ct}$] \phantom{sdd}
    
    Compute $\rho_\ct \leftarrow \QFE.\Enc(\mpk, (\rho_\msg \otimes |0\rangle\langle0|^{\otimes (n+2l+t+s)} \otimes |0\rangle\langle0|))$

    Output $\rho_\ct$.

    \item[$\bm{\Dec(\sk_C, \rho_\ct) \rightarrow \rho_\msg}$] 
    
    Run $\QFE.\Dec(\sk_C, \rho_\ct) = \rho_\msg$ and output $\rho_\msg$.
\end{itemize}

\begin{theorem}\label{qfe:thm:UEQFE}
    Any single-query QFE scheme for n-qubit messages and universal circuits (\cref{qfe:def:indsecurity}) is a single-query unclonable-indistinguishable functional encryption scheme (\cref{qfe:def:UEQFE}) if there exists an unclonable-indistinguishable encryption scheme with quantum decryption keys for single bit messages (\cref{qfe:def:ueqdec}). 
\end{theorem}

\begin{proof}\phantom{sdasd}
    \paragraph{Correctness}
    The scheme is correct based on the correctness of the underlying functional encryption scheme.

    We show security by a series of Hybrids:
    \paragraph{Hybrid 0:} This is the unclonable functional encrpytion experiment $\Exp^{QFE-UE-IND}_{\Adv}$.

    \paragraph{Hybrid 1:} In this Hybrid we change how the challenge ciphertext is created, in particular we change the flag bit $f$ to 1 such that the circuit executes the second case of it's description. 

    $\Enc^*(\mpk, \rho_{\msg_0}, \rho_{\msg_1}):$
    \begin{enumerate}[align=left, leftmargin=2.8em]
        \item Run $\UEQ.\Keygen(1^\lambda, r^*) = (\ek, |\dk_0\rangle)$. Produce another copy of the decryption key by using the same randomness $\UEQ.\Keygen(1^\lambda, r^*) = (\ek, |\dk_1\rangle)$.
        \item Sample 2 sets of $l(\lambda)$ EPR pairs $\sigma_0^{AB}$ and $\sigma_1^{AB}$. Let $\sigma_0^A$, $\sigma_1^A$ denote  registers containing the first qubit of each EPR pair and $\sigma_0^B$,$\sigma_1^B$ denote registers containing the second qubit of each EPR pair respectively.

        \item Sample $b \leftarrow \bin$.
        \item Run $\rho_{UE} \leftarrow \UEQ.\Enc(1^\lambda, b)$.
        \item Create the ciphertext 
        $$\rho_\ct = \QFE.\Enc(\mpk, (\rho_{\msg_0} \otimes \rho_{\msg_{1}} \otimes \sigma_0^A \otimes \sigma_0^B \otimes \rho_{UE} \otimes C_{Dec} \otimes |1\rangle\langle 1|))$$
        \item Teleport the key ($0^\lambda \otimes |\dk_0\rangle$),$(0^\lambda \otimes |\dk_1\rangle)$ through the EPR pairs $\sigma_0^B, \sigma_1^B$ respectively and obtain the teleportation keys $(a_0',b_0'),(a_1',b_1')$.
        Output $(\rho_{\ct}, (a_0', b_0'), ( a_1', b_1'))$.
    \end{enumerate}

\begin{claim}
    $|p_0-p_1| \leq negl(\lambda)$ where $p_0$ is the winning probability of the adversary in Hybrid 0 and $p_1$ is the winning probability in Hybrid 1.
\end{claim}

\begin{proof}
    We show that an adversary $\Adv = (A, B, C)$ that can win in Hybrid 0 with a higher probability than in Hybrid 1 can be used to break IND-security of the underlying $\QFE$ scheme.
    
    During the reduction both parties $B$ and $C$ will need to obtain independently sampled secret keys for the functional encryption scheme. Since our $\QFE$ scheme is only single-query secure we cannot allow the adversary to sample two secret keys. Instead we reduce to the notion of 2-player single-query IND-security which we define in \cref{qfe:def:na2playerindsecurity}. This security notion allows two recipients of a ciphertext that don't further communicate to each receive a functional secret key from the single-query secure $\QFE$ scheme. We also show that this security notion is implied by single-query IND-secure $\QFE$.

    Let $\Adv^* = (A^*,B^*,C^*)$ be the adversary in the 2-player single-query non-adaptive IND-security game against the quantum functional encryption scheme. $A^*$ receives the public key $\mpk$ from the experiment and runs $A$ on input $(1^\lambda, \mpk)$ until $A$ outputs messages $\rho_{\msg_0},\rho_{\msg_1}$. Sample $b \leftarrow \bin$.
    
    To create the first challenge message $\rho_{\msg^*_0}$ $A^*$ performs the steps of the honest $\Enc$ algorithm without the creation of the $\QFE$ ciphertext. Then $A^*$ sets 

    $$\rho_{\msg^*_0} = (\rho_{\msg_b}\otimes |0\rangle\langle0|^{n(\lambda)} \otimes |0\rangle\langle0|^{\otimes 2l(\lambda)+ t(\lambda)+s(\lambda)}\otimes |0\rangle\langle0|)$$

    To create the challenge message $\rho_{\msg^*_0}$ $A^*$ performs encryption as defined in $\Enc^*$ without the creation of the ciphertext (step 5) but with the teleportation (step 6) to obtain teleportation keys $(a_0,b_0), (a_1,b_1)$. In step 4 use the bit $b$ that was already sampled. Then the message $\rho_{\msg^*_1}$ is defined as

    $$\rho_{\msg^*_1} = (\rho_{\msg_0}\otimes \rho_{\msg_1} \otimes \sigma_0^A \otimes \sigma_0^B \otimes \rho_{UE} \otimes C_{Dec} \otimes |1\rangle\langle 1|)$$

    Note that $\Adv$ is not required to copy the messages $\rho_{\msg_0}$ and $\rho_{\msg_1}$ to define the challenge messages. According to the IND-security experiment $\Adv^*$ can define both messages by referring to a single quantum state, this is a special case of \cref{qfe:def:admis}. $\Adv^*$ declares the messages $\rho_{\msg^*_0},\rho_{\msg^*_1}$ and additionally outputs the circuit descriptions $\sk_{C_B} = U_{(C,a_0,b_0)}$ and $\sk_{C_C} = U_{(C, a_1,b_1)}$.

    Both $\sk_{C_B}$ and $\sk_{C_C}$ are admissible function queries since 

    $$ U_{(C,a_0,b_0)}(\rho_{\msg^*_0}) = \rho_{\msg_b} = U_{(C,a_0,b_0)}(\rho_{\msg^*_1})$$
    
    and 
    
    $$ U_{(C,a_1,b_1)}(\rho_{\msg^*_0}) = \rho_{\msg_b} = U_{(C,a_1,b_1)}(\rho_{\msg^*_1})$$

    and $\Adv^*$ is no longer entangled with the input messages.

    $\Adv^*$ receives the ciphertext $\rho_\ct$ and runs $\Adv$ to obtain $\rho_{BC}$. 

    Then $B^*$ and $C^*$ are activated with the state $(\rho_B,\sk_{C_B},b)$ and $(\rho_C, \sk_{C_C},b)$ respectively and each run $B$ and $C$ on input $(\rho_B,\sk_{C_B})$ and $(\rho_C,\sk_{C_C})$ respectively until they output a bit $b_B$, $b_C$. 

    $A^*,B^*,C^*$ simulate Hybrid 0 if $\rho^*_{\msg_0}$ is picked as challenge and they simulate Hybrid 1 if $\rho^*_{\msg_1}$ is picked. Let $b^* \in \bin$ denote the choice of the challenge message. 
    
    $B^*$ outputs $b_B^* = 0$ if $b_B = b$ otherwise $B^*$ outputs $b_B^* = 1$, similarly $C^*$ outputs $b_C^* = 0$ if $b_C = b$ and otherwise outputs $b_C^* = 1$.

    This means that if $b_B = b_C = b$ we have $b_B^* = b_C^* = 0$. 
     For Hybrid 1 we show in \cref{qfe:lem:reductiontoueq} that $b_B = b_C = b$ only occurs with negligible advantage.

    The winning probability of $\Adv^*$ is 
    \begin{align*}
        &\frac{1}{2} (\Pr[b_B^* = b_C^* = 0 | b^* = 0] + \Pr[b_B^* = b_C^* = 1 | b^* = 1])\\
        &= \frac{1}{2} (\Pr[b_B = b_C = b | b, b^* = 0] + (\Pr[b_B \neq b \lor b_C \neq b | b, b^* = 1 ]) )\\ 
        &=\frac{1}{2} ( \underbrace{\Pr[b_B = b_C = b | b, b^* = 0]}_{p_0} + (1 - \underbrace{\Pr[b_B = b_C = b | b, b^* = 1 ]}_{1/2 + negl(\lambda)}) )\\ 
        &= \frac{1}{2} (p_0 + 1/2-negl(\lambda))
    \end{align*}

    Therefore, if the advantage of $\Adv$ in Hybrid 0 is non-negligible $p_0 = \frac{1}{2} + non-negl(\lambda)$, $\Adv^*$ can break the 2-player IND-security of QFE with non-negligible advantage. 
    
\end{proof}

\begin{lemma}
\label{qfe:lem:reductiontoueq}
        In Hybrid 1 the advantage of $\Adv = (A, B, C)$ is negligible if UEQ is secure.
    \end{lemma}
\begin{proof}
    We show that an adversary that breaks the security of the unclonable functional encryption scheme breaks uncloneability of the underlying $\UEQ$ encryption scheme with quantum decryption keys with the same advantage. 

    Let $\Adv = (A,B,C)$ be an adversary that breaks the security of the unclonable functional encryption scheme. Then we can build an adversary $\Adv^* = (A^*, B^*, C^*)$ that breaks the security of the $\UEQ$ scheme. In the role of $A^*$ send challenge messages $b_0  = 0, b_1= 1$ to the experiment and obtain $\rho_{UE}$.
    
    Create the setup for unclonable functional encryption $(\mpk, \msk) \leftarrow \Setup(1^\lambda)$ and run $A$ on input $(1^\lambda, \mpk)$. Receive the challenge messages $\rho_{\msg_0}, \rho_{\msg_1}$ from $A$.
    
    Build the ciphertext as in Hybrid 2: Sample 2 sets of n EPR pairs $\sigma_0^{AB}$ and $\sigma_1^{AB}$. Let $\sigma_0^A$, $\sigma_1^A$ denote registers containing the first qubit of each EPR pair and $\sigma_0^B$,$\sigma_0^B$ denote registers containing the second qubit of each EPR pair respectively. 
    
    Create the ciphertext $\rho_\ct = \QFE.\Enc(\mpk,(\rho_{\msg_0}\otimes \rho_{\msg_1} \otimes \sigma_0^A \otimes \sigma_0^B \otimes \rho_{UE}  \otimes C_{Dec} \otimes |1\rangle\langle 1|))$ and send $\rho_\ct$ to  $A$. If $\rho_{UE}$ is an encryption of $b=0$ then $\rho_\ct$ is an encryption of $\rho_{\msg_0}$, if $\rho_{UE}$ is an encryption of $b=1$ then $\rho_\ct$ is an encryption of $\rho_{\msg_1}$.

    Additionally $\Adv$ outputs $\rho_{\st_B} = (\msk, \sigma_0^B)$ and $\rho_{\st_C} = (\msk,\sigma_1^B)$.
    
    $A$ performs the splitting channel and outputs a state $\rho^{BC}$ which is also the state that $A^*$ defines as it's result of the splitting channel. 
    Now $B^*$ and $C^*$ are activated. They take as input the states $\rho^B,\rho_{\st_B} $ and $\rho^C, \rho_{\st_C}$ respectively and each receive a copy of the secret key $|\dk\rangle$ from the experiment.

    $B^*$ teleports the state $ 0^\lambda \otimes |\dk\rangle$ trough the EPR pairs $\sigma_0^B$ and obtains the teleportation correction keys $a_0,b_0$. 
    
    He produces the secret key $\dk_B = \QFE.\Keygen(\msk, U_{(C,a_0,b_0)})$ where $C$ is the identity circuit. 
    He runs the adversary $B$ on input $\rho^B$ and the secret key $\dk_B$ and outputs whatever $B^*$ outputs.

    $C^*$ does the same actions as $B^*$ on his respective EPR pairs. He teleports the state $0^\lambda \otimes |\dk\rangle$ trough the EPR pairs $\sigma_1^B$ and obtains the teleportation correction keys $a_1,b_1$. 
    
    He produces the secret key $\dk_C = \QFE.\Keygen(\msk, U_{(C, a_1,b_1)})$ where $C$ is the identity circuit. He runs the adversary $C$ on input $\rho^C$ and the secret key $\dk_C$ and outputs whatever $C^*$ outputs. 
    
    $(A^*, B^*, C^*)$ wins with the same probability as $(A,B,C)$. 
\end{proof}
    
\end{proof}

\begin{lemma}
    Any non-adaptive unclonable-indistinguishable functional encryption scheme is a  public key unclonable-indistinguishable encryption scheme with variable decryption keys  (\cref{qfe:def:pkeue}). 
\end{lemma}
\begin{proof}
    Note that the key queries in the unclonable functional encryption experiment do not have to be admissible queries. In particular $B$ and $C$ can both obtain a secret key for the circuit that computes the identity even if $\rho_{\msg_0}, \rho_{\msg_1}$ are different messages. This defines decryption keys for an unclonable public-key encryption scheme. Security and correctness follow as a special case of the security and correctness of the unclonable functional encryption scheme.
\end{proof} 

\begin{corollary}
    There exists a universal public-key unclonable-indisinguishable encryption scheme with variable decrpytion keys  (\cref{qfe:def:pkeue})  for n-bit messages assuming a single-query QFE scheme (\cref{qfe:def:UEQFE}) and assuming the existance of an unclonable-indistinguishable encryption scheme with quantum decryption keys for single bit messages (\cref{qfe:def:ueqdec}). 
\end{corollary}

\begin{remark}
   The techniques given in this section can be applied more generally. Assuming the existence of UE with a given level of security, any QFE scheme can be modified to achieve unclonable security at the same level, without explicitly using the original UE scheme. For example, instead of reducing to the security of UE with variable decryption keys we could reduce to an ideal UE scheme, i.e. an unclonable-indistinguishable encryption scheme with standard classical decryption keys and then also achieve this ideal security notion for the unclonable QFE scheme. 
\end{remark}
\section{From Quantum Multi-input Functional Encryption to Quantum Indistinguishability Obfuscation}
\label{qfe:sec:qmife}
In this section we first define multi-input functional encryption in the quantum setting. Then, we show that multi-input quantum functional encryption implies quantum indistinguishability obfuscation. In the classical setting it is known that IND-secure multi-input functional encryption and qiO are equivalent, one notion can be constructed from the other~\cite{EC:GGGJKH14}. An interesting open question that we do not address in this work is from what assumptions IND-secure quantum multi-input functional encryption could be constructed.

\subsection{Definitions}
 In this section we are switching to a secret-key flavor of functional encryption. The adversary cannot create ciphertexts on its own but has to query the $\Enc$ functionality for this. First we establish the syntax of a quantum multi-input functional encryption scheme.

A quantum multi-input functional encryption scheme $\QMIFE$ for a family of circuits $\{\mathcal{C}_\lambda\}_{\lambda \in \N}$ with input space $\mathcal{X}_\lambda$ and output space $\mathcal{Y}_\lambda$ consists of four algorithms $(\Setup, \Keygen, \Enc, \Dec)$ as described below.
\begin{itemize} [align=left,leftmargin=2.8em]
    \item[$\bm{\Setup}$]  $\Setup(1^\lambda, n) \rightarrow (\msk, \ek_1, \ldots, \ek_n)$ is a QPT algorithm that takes as input the security parameter $\lambda \in \N$ and the number of input qubits $n \in \N$. It outputs $n$ encryption keys $\ek_1, \ldots, \ek_n$ and a master secret key $\msk$.
    \item[$\bm{\Keygen}$] $\Keygen(\msk, C)\rightarrow \sk_C$ is a QPT algorithm that takes as input the master secret key $\msk$ and a circuit $C \in \mathcal{C}_\lambda$ and outputs a corresponding secret key $\sk_C$.
    \item[$\bm{\Enc}$] $\Enc(\ek,\rho_x)\rightarrow \rho_\ct$ is a QPT algorithm that takes as input an encryption key $\ek_i \in\left(\ek_1, \ldots, \ek_n\right)$ and an input message $\rho_x \in \mathcal{X}$ and outputs a ciphertext $\rho_\ct$. In the case where all of the encryption keys $\ek_i$ are the same, we assume that each ciphertext $\rho_\ct$ has an associated label $i$ to denote that the encrypted plaintext constitutes an $i$'th input the circuit $C \in \mathcal{C}_\lambda$. For convenience of notation, we omit the labels from the explicit description of the ciphertexts. It might also be useful to distinguish between classical and quantum input. Since any classical input can be embedded in a quantum state we do not explicitly differentiate between these two types of inputs here. 
    \item[$\bm{\Dec}$] $\Dec(\sk_C, \rho_{\ct_1}, \ldots, \rho_{\ct_n})\rightarrow \rho_y$ is a QPT algorithm that takes as input a secret key $\sk_C$ and $n$ ciphertexts $\rho_{\ct_1}, \ldots, \rho_{\ct_n}$ and outputs a state $\rho_y \in \mathcal{Y}_\lambda$.
\end{itemize}

In the description of this syntax we only declared inputs, outputs and ciphertexts explicitly as quantum states but other parts of the scheme such as keys could also contain quantum data in a specific instantiation. 

\subsubsection{Indistinguishability Based Security}The scheme is parameterized by $k$ which denotes the number of ciphertexts the adversary is allowed to learn per secret key.

Admissible challenge messages are defined using the same concept as in \cref{qfe:sec:inddef} for the IND-security of simple functional encryption. We additionally have to take into account that the adversary can choose between different combinations of input ciphertexts to evaluate the circuit. 

\begin{definition}(Admissible queries for $\QMIFE$)
\label{qfe:def:compatabilityQ}
Let $Q$ be a set of circuits containing circuits $C \in \{C_\lambda\}_{\lambda \in \N}$ with input size $n$. The adversary in $\Exp_{\mathcal{A}}^{\operatorname{IND}-\QMIFE}$ specifies a challenge query by states $\rho_{m^0}, \rho_{m^1}$ with the following structure: A state $ \rho_{m^b_{h,j}}$ is defined by taking the partial trace of $\rho_{m^b}$ indexed by $h \in [n], j \in [k]$:

$$ \rho_{m^b_{h,j}} = \Tr_{(\Bar{h},\Bar{j})}[\rho_{m^b}]$$

The messages are grouped in vectors $X^0, X^1$ where $X^b = \{\rho_{m^b_{1,j}}^{EU}, \cdots, \rho_{m^b_{n,j}}^{EU}\}_{j \in [k]}$. For each challenge message indexed by $h \in [n], j \in [k], b \in \bin $ the adversary can specify a register  $E$ that will be used for encryption and a register $U$ that will be returned unencrypted. The challenge messages corresponding to $1-b$ are not returned to the adversary.

Let $\rho^U_{m^b_{j^*}}$ be a state that groups together the registers not used for encryption, the state contains $\rho^{U}_{m^b_{h,j^*}}$ for all $h \in [n]$ and a specific choice of $j^* = (j_1, \ldots, j_n)$ with each $j_i \in [k]$.

 We say $(X^0,X^1)$ and $Q$ are compatible  if the following property is satisfied for all $C \in Q$ and for all choices of $j^*$:

\begin{align*}
     \T\left(\sum_{i} p_i C(\rho^E_{\msg^0_{1,j_1,i}}, \cdots, \rho^E_{\msg^0_{n,j_n,i}}) \otimes \rho^U_{m^b_{j^*,i}} \otimes \rho_{A_i}, \sum_{i} q_i C(\rho^E_{\msg^1_{1,j_1,i}}, \cdots, \rho^E_{\msg^1_{n,j_n,i}}) \otimes \rho^U_{m^b_{j^*,i}} \otimes \rho_{A_i}\right) \leq negl(\lambda)
\end{align*}

where $\rho_A$ is the local state of the adversary.

\end{definition}

\begin{definition}(Quantum MIFE IND-Security)

Let $\QMIFE = (\Setup,\Keygen, \Enc, \Dec)$ be a quantum multi-input functional encrpytion scheme for a circuit family $\{C_\lambda\}_{\lambda \in \N}$ and let $\Adv = (\Adv_1, \Adv_2)$ be a QPT adversary.

\begin{align*}
& \Exp_{\mathcal{A}}^{\operatorname{IND}-\QMIFE}\left(1^\lambda\right): \\
& (\ek, \msk) \leftarrow \Setup\left(1^\lambda,n\right) \\
& \left(\vec{X}^0, \vec{X}^1, \rho_{\st_1}\right) \leftarrow \mathcal{A}_1^{\Keygen\left(\msk_{,}\cdot\right)}\left(1^\lambda,n\right) \text { where } \vec{X}^{\ell}=\left\{\rho_{\msg_{1, j}}^{\ell}, \ldots, \rho_{\msg_{n, j}}^{\ell}\right\}_{j \in [k]} \\
& b \leftarrow\{0,1\} \\
& \ct_{i, j} \leftarrow \Enc\left(\mathrm{ek}_i, \rho_{\msg_{i, j}}^b\right) \forall i \in[n], j \in[k] \\
& b^{\prime} \leftarrow \mathcal{A}_2^{O(\cdot)}\left(\rho_{\st_1}, \{\rho_{\ct_{i,j}}\}_{i \in [n], j \in [k]}\right) \\
& \text{Output: } \left(b=b^{\prime}\right)
\end{align*}

Let $Q$ denote the entire set of key queries made by $\mathcal{A}$. Then, the challenge message vectors $\vec{X}_0$ and $\vec{X}_1$ chosen by $\mathcal{A}_1$ must be compatible with $Q$ (\cref{qfe:def:compatabilityQ}).
The scheme is $k$-IND-secure if for every QPT adversary $\mathcal{A}=\left( \mathcal{A}_1, \mathcal{A}_2\right)$, the advantage of $\mathcal{A}$ defined as
$$
\operatorname{Adv}_{\mathcal{A}}^{\QMIFE, \mathrm{IND}}\left(1^\lambda\right)=\left|\operatorname{Pr}\left[\Exp_{\mathcal{A}}^{\operatorname{IND}-\QMIFE}\left(1^\lambda\right)=1\right]-\frac{1}{2}\right| \leq negl(\lambda)
$$

Adaptive vs. Non-adaptive security

\begin{itemize}
    \item The scheme is called non-adaptively secure if the the adversary only queries
the $\Keygen$ oracle before receiving a ciphertext. Then the oracle $O(\cdot)$ is the
empty oracle.
\item The scheme is called adaptively secure if the adversary can either query the
KeyGen oracle before or after receiving the ciphertext. Then the oracle $O(\cdot)$
is the function $\Keygen(\msk, \cdot)$.
\end{itemize}

\end{definition}

\subsubsection{Simulation Security} In the simulation security setting we need to give the simulator access to the output of the circuit evaluated on any combination of inputs. In the classical setting this is simple: There is a trusted part which holds the input messages $\vec{X}=\left\{\msg_{1, j}, \ldots, \msg_{n, j}\right\}_{j \in [k]}$ and the simulator can specify a queries of the form $(g, j_1, \ldots, j_n)$ where $g$ is a function and $j_1$ to $j_n$ are indices selecting the input for the function. The simulator can make multiple queries using an arbitrary combination of indices and any function that the adversary requested keys for. 

In the quantum setting we run into the issue that the inputs which are quantum states cannot be reused arbitrarily. On the other hand for some functionalities it might be possible or even desired that after obtaining one output the state of the input ciphertext can be restored by uncomputing the decryption unitary. Then the inputs can be reused to evaluate the same or a different functionality on a combination of input ciphertexts. 

In the quantum setting a standard way of modelling quantum access to a oracle is the following. The user specifies a query $|\phi\rangle = \sum_i \alpha_i |x_i\rangle |u_i\rangle$ and the oracle answers with the state  $|\phi'\rangle = \sum_i \alpha_i |x_i\rangle |u_i \xor f(x_i)\rangle$. This state is computed by first applying $f$ to the $x$ register, xoring the result to the $u$ register and uncomputing the function on the $x$ register. We can use the same concept to define how the trusted party answers queries with the difference that the trusted party already holds the input register. This allows the trusted party to reuse the input messages and answer multiple queries of the form $(g, \sum \alpha_l |j_{1,l}\otimes \ldots \otimes j_{n,l}\rangle )$. It is to be noted though that this causes the answer register to be entangled with the input register. Therefore a measurement by the simulator will also collapse the input state and multiple evaluations are not guaranteed to work correctly.

\begin{definition}(Quantum MIFE SIM-Security)
    A multi-input functional encryption scheme for a circuit family $\{C_\lambda\}_{\lambda\in \N}$ is k-SIM-secure if for every QPT adversary $\Adv = (\Adv_1, \Adv_2)$ there exists a stateful simulator $\Sim$ such that the outputs of the following experiments are computationally indistinguishable: 

     \begin{table}[H]
        \centering
        \begin{tabular}{p{5.5cm}|p{6.5cm}}
            $\Exp_{\Adv}^{Real}(1^\lambda)$ & $\Exp_{\Adv}^{Ideal}(1^\lambda)$ \\
            $(\{\ek_i\}_{i\in[n]},\msk) \leftarrow \Setup(1^\lambda,n)$& \\
            $ (X,  \st) \leftarrow \Adv_1^{\Keygen(\cdot)}(1^\lambda,n)$ & $ (X,  \st) \leftarrow \Adv_1^{O_1(\cdot)}(1^\lambda,n)$\\
            where $X = \{ {\rho_{\msg_{1,j}}}, \dots, {\rho_{\msg_{n,j}}} \}_{j \in [k]}$&where $X = \{ {\rho_{\msg_{1,j}}}, \dots, {\rho_{\msg_{n,j}}} \}_{j \in [k]}$\\
            $ {\rho_{\ct_{i,j}}} \leftarrow \Enc(\ek_i,  {\rho_{\msg_{i,j}}})  \forall i \in [n], j \in [k]$& $ \{\rho_{\ct_{i,j}}\}_{i, j} \leftarrow \Sim^{\TP(\cdot)}(1^\lambda, 1^{|C|}, \{1^{| \rho_{\msg_{i,j}}|}\}_{i\in[n], j \in [k]})$ \\
            $\alpha \leftarrow \Adv_2^{O_2'(\cdot)}( \{\rho_{\ct_{i,j}}\}_{i\in[n],j\in[k]},  \st)$ &  $\alpha \leftarrow \Adv_2^{O_2(\cdot)}( \{\rho_{\ct_{i,j}}\}_{i\in[n],j\in[k]},  \st)$\\
            The experiment outputs $\alpha$ & The experiment outputs $\alpha$\\
        \end{tabular}
    \end{table}

    where the oracle $\TP(\cdot)$ denotes the ideal world trusted party. It accepts queries of the form $(g, \sum \alpha_l |j_{1,l}\otimes \ldots \otimes j_{n,l}\rangle )$ and computes 

     \begin{align*}
        &\sum \alpha_{l} |j_{1,l}, \cdots, j_{n,l}\rangle \otimes \rho_\msg \otimes g(\rho_{\msg_{1,j_{1,l}}}, \ldots,\rho_{\msg_{n,j_{n,l}}})\\
    \end{align*}

    The message register $\rho_\msg$ is kept by $\TP$ and used for future queries, the rest is  returned to the simulator.

    The oracle $O_1(\cdot)$ is a $\Keygen$ oracle controlled by the simulator and the oracle $O_2(\cdot)$ is a $\Keygen$ oracle controlled by the simulator with access to the trusted party $\TP$. A simulator is admissible if it only queries the trusted party on functionalities that $\Adv$ queried to its oracle.
\end{definition}

\begin{remark} 
    In this definition we have for the first time in this work considered the case of multiple function queries. We remark that defining a multi-query QFE scheme for only a single ciphertext runs into the same issues we described above. Given multiple  function keys a single ciphertext has the possibility to be evaluated to different outputs but physically not all these evaluations might be possible. Therefore a solution as presented here for the multi-input case is necessary and a definition for a multi-query simulation secure $\QFE$ scheme can be derived from this definition by restricting the input to a single message $n=1$.
\end{remark}

\subsection{IND-secure QMIFE implies qiO}
\begin{theorem}
\label{qfe:thm:indtoqio}
    A $\QMIFE$ scheme that fullfills non-adaptive single-query 2-IND-security unconditionally implies quantum indistinguishability obfuscation.
\end{theorem}

\begin{proof}
    Let $\QMIFE$ be a quantum multi input functional encryption scheme. We define an obfuscation scheme $(\Obf, \Eval)$ for a familiy of circuits $\{C_\lambda\}_{\lambda \in \N}$ that take as input $n$ qubits and are described by a classical string of length $l$.
    
    \paragraph{$\Obf(C)$:} 
    \begin{itemize}
        \item Run $\QMIFE.\Setup(1^\lambda, n')\rightarrow (\msk, \ek_1, \ldots, \ek_{n'})$ where $n' = 3n + l$
        \item Run $\QMIFE.\Keygen(U, \msk) \rightarrow sk_U$ where $U$ is a variant of a universal circuit that computes $U(C, \rho_1, \cdots, \rho_{n}, a_1, b_1, \ldots, a_{n}, b_{n}) = C( X^{a_1} Z^{b_1} \rho_1, \cdots, X^{a_{n}} Z^{a_{n}} \rho_{n})$
        \item Create $n$ ciphertexts that encrypt the bit $b=0$ and $n$ ciphertexts that encrypt $b=1$:
        $$ \forall i \in [2n], b \in \bin : {\ct^b_i} \leftarrow \QMIFE.\Enc(\ek_{i}, b)$$
         \item Create $n$ EPR pairs and take the first qubit of each EPR pair $\rho_e = (\rho_{e_1}, \rho_{e_2})$ and encrypt it: 
        $$\forall i \in [n] : \rho_{ct_{2n+i}} \leftarrow \QMIFE.\Enc(\ek_{2n+i}, \rho_{e_{i,1}})$$.
        \item Encrypt the circuit $C$: 
        $$ \ct_C \leftarrow \QMIFE.\Enc(\ek_{3n+1}, C)$$
    \item Output $\Tilde{C} = (\sk_U, \ct_C, \{\ct_i^b\}_{i \in [n],b\in \bin}, \{\rho_{\ct_i}\}_{i \in [n]}, \{ \rho_{e_{i,2}}\}_{i \in [n]})$
    \end{itemize}

    \paragraph{$\Eval(\Tilde{C}, \rho_x)$}
    \begin{itemize}
        \item Teleport the state $\rho_x$ which is of size $n$ trough the EPR pairs $\rho_{e_{1,2}} \otimes \cdots \otimes \rho_{e_{n,2}}$ and obtain $((a_1, b_1), \cdots, (a_n, b_n))$ as teleportation keys.
        \item Select the remaining ciphertexts such that they are encryptions of $(a_i,b_i)$: $\forall i \in [n]$ select $\ct_i^{a_i}$ and $\ct_{i+1}^{b_i}$.
        \item Run $\QMIFE.\Dec(\sk_U, \ct_C, \rho_{\ct_1}, \ldots, \rho_{\ct_n}, \ct^{a_1}_1, \ct^{b_1}_2, \cdots, \ct^{a_n}_{2n-1}, \ct^{b_n}_{2n}) = \rho_y$
    \end{itemize}

First we analyse the correctness of the scheme. By correctness of the $\QMIFE$ scheme and correctness of the teleportation gadgets the scheme outputs the correct evaluation.
\begin{align*}
    &\QMIFE.\Dec(\sk_U, \ct_C, \rho_{\ct_1}, \ldots, \rho_{\ct_n}, \ct^{a_1}_1, \ct^{b_1}_2, \cdots, \ct^{a_n}_{2n-1}, \ct^{b_n}_{2n})\\
    &=U(C, \rho_1, \cdots, \rho_{n}, a_1, b_1, \ldots, a_{n}, b_{n})\\
    &=C( X^{a_1} Z^{b_1} \rho_1, \cdots, X^{a_{n}} Z^{a_{n}} \rho_{n})\\
    &= C(\rho_{x_1}, \ldots, \rho_{x_n})
\end{align*}

We note that a honest user will only be guaranteed one use of the obfuscated program since the teleportation ciphertexts are consumed during this operation. If the quantum circuit belongs to a class of circuits that only take classical inputs we can avoid the use of the teleportation helper state and the scheme can be redefined to let the user select it's classical inputs in the same manner as the bits $(a_i,b_i)$ are selected here. This will still not guarantee a reusable qiO scheme since the obfuscated circuit itself might be a quantum state that collapses during evaluation. 

No we show that the security of the qiO scheme can be reduced to the security of the underlying $\QMIFE$ scheme. Let $\Adv$ be an adversary that wins the qiO experiment with non-negligible advantage. Then we can construct an adversary $\AdvB$ that wins the $\QMIFE$ IND-security experiment with non-negligible advantage. 

$\AdvB$ receives $1^\lambda$ and runs $\Adv$ on input $1^\lambda$ until $\Adv$ outputs $(C_0, C_1)$.  
$\AdvB$ queries the $\Keygen$ oracle on the function $U$ as defined above and receives the secret key $\sk_U$.

$\AdvB$ constructs it's challenge vectors as follows: 
Sample $n$ EPR pairs $\rho_{e_i} = \frac{1}{\sqrt{2}} (|0\rangle^1|0\rangle^2 \otimes |1\rangle^1|1\rangle^2) = (\rho_{e_{i,1}},\rho_{e_{i,2}}) $ and put the first qubit each in the challenge vector $X^0$ and put the second qubit each in the 'do not encrypt' part of the challenge query. 
$X^0 = (C_0, \rho_{e^0_{1,1}}, \cdots, \rho_{e^0_{n,1}}, \{a_{i,1} = 0, a_{i,2} = 1, b_{i,1} = 0, b_{i,2} = 1\}_{i \in [n]})$
Sample $n$ additional  EPR pairs and put the first qubit each in the challenge vector $X^1$ and put the second qubit each in the 'do not encrypt' part of the challenge query. 
$X^1 = (C_1, \rho_{e^1_{1,1}}, \cdots, \rho_{e^1_{n,1}}, \{a_{i,1} = 0, a_{i,2} = 1, b_{i,1} = 0, b_{i,2} = 1\}_{i \in [n]})$. Let $\rho^U_{X^b} = \rho_{e^b_{1,2}}\otimes \cdots \otimes \rho_{e^b_{n,2}}$ for each $b \in \bin$.

The experiment sends $(\ct_C,  \{\rho_{\ct_i}\}_{i \in [n]}, \{\ct_i^d\}_{i \in [n],d\in \bin}, \{ \rho_{e_{i,2}}\}_{i \in [n]})$ where $\ct_C$ is the encryption of $C_b$, $\{\rho_{\ct_i}\}_{i \in [n]}$ are the encryptions of the EPR pair halves, $\{\ct_i^d\}_{i \in [2n],d\in \{a,b\}}$ are the encryptions of $a_i,b_i$ and the unencrypted second halves of the EPR pairs $\{ \rho_{e_{i,2}}\}_{i \in [n]}$ to $\AdvB$.

No we need to verify that the query $(U,X^0, X^1)$ forms an admissible query for the $\QMIFE$ IND-experiment according to \cref{qfe:def:compatabilityQ}.
The challenge vectors $X^1, X^0$ only differ in the first component which contains $C_b$. Let the state $\rho_{X^0,ab}$ and $\rho_{X^1,ab}$ denote the state containing the classical bit queries of each challenge vector. Then for the inputs to $U(C_0, \cdot), U(C_1, \cdot)$ it holds that 

$$ \T\left(\underbrace{\rho_{e^0_{1,1}}\otimes  \cdots \otimes \rho_{e^0_{n,1}}}_{\rho_{X^0,e}} \otimes \rho_{X^0,ab} , \underbrace{\rho_{e^1_{1,1}}\otimes \cdots \otimes \rho_{e^1_{n,1}}}_{\rho_{X^1,e}} \otimes \rho_{X^1,ab}  \right) = 0 $$

$\AdvB$ does not need to keep any information other than the secret key in it's local state $\rho_B$, in particular $\AdvB$ is not entangled with any part of the challenge query (the remaining halves of the EPR pairs of the challenge query are given away to the experiment and returned without encryption per the definition of admissible queries). Let these qubits be contained in the registers $\rho^U_{X^0}$ and $\rho^U_{X^1}$ respectively.

$$\T\left( \sum_i \rho_{X^0,e,i} \otimes \rho_{X^0,ab} \otimes \rho^U_{X^0,i} \otimes \rho_B, 
    \sum_i \rho_{X^1,e,i}  \otimes \rho_{X^1,ab} \otimes \rho^U_{X^1,i}\otimes \rho_B\right) =0$$

By the requirement of the qiO IND-experiment the circuits $C_0, C_1$ are perfectly functionally equivalent. The circuits $U(C_0, \cdot), U(C_1, \cdot)$ inherit this property. 

Then, 

    $$\T\left( \sum_i U(C_0,\rho_{X^0,e,i}, \rho_{X^0,ab}) \otimes \rho^U_{X^0,i} \otimes \rho_B, 
    \sum_i U(C_1,\rho_{X^1,e,i},\rho_{X^1,ab}) \otimes \rho^U_{X^1,i}, \rho_B\right) =0$$
which means the query $(U,X^0, X^1)$ forms an admissible query for the $\QMIFE$ IND-experiment.

$\AdvB$ sends the obfuscated circuit $\Tilde{C_b} = (\sk_U,\ct_C, \{\ct_i^d\}_{i \in [2n],d\in \{a,b\}}, \{\rho_{\ct_i}\}_{i \in [n]}, \{ \rho_{e_{i,2}}\}_{i \in [n]})$ to $\Adv$.
$\AdvB$ outputs whatever $\Adv$ outputs. If the $\QMIFE$ IND-experiment selected $X^0$ as a challenge it perfectly simulates an obfuscation of $C_0$ if the $\QMIFE$ IND-experiment selected $X^1$ as a challenge it perfectly simulates an obfuscation of $C_1$.
Therefore $\AdvB$ wins with the same probability as $\Adv$.

\end{proof}

\subsection{SIM-secure QMIFE implies QVBB}
 In this section we show that a simulation-secure QMIFE implies QVBB, even if we cannot hope to achieve such a construction. It is known that quantum virtual black box obfuscation is impossible to achieve for general circuits~\cite{AF16}, therefore, impossibility of QMIFE immediately follows.

\begin{theorem}\label{QMIFE:thm:VBBO}
    A QMIFE scheme that fulfills non-adaptive single-query 2-SIM-security unconditionally implies virtual black box quantum obfuscation.
\end{theorem}

\begin{proof}
The same construction as in the previous proof of \cref{qfe:thm:indtoqio} implies QVBB if the QMIFE scheme is 2-SIM secure. 

    For any adversary $\Adv(1^\lambda)$ we define a simulator $\Sim(1^\lambda)$ for the scheme as follows. Let $\widetilde{\Sim}(1^\lambda)$ be the simulator for the QMIFE scheme. Then $\Sim$ creates n EPR pairs as required by the construction and 
runs the simulator $\widetilde{\Sim}$ to create the remaining parts of the obfuscated circuit, i.e. the ciphertexts and the key for the universal circuit pairs as defined in \cref{qfe:thm:indtoqio}. Upon receiving a query from $\widetilde{\Sim}$ the simulator forwards the query to its own oracle. Indistinguishability follows from the security of the QMIFE scheme. 
\end{proof}

\printbibliography

\appendix
\section{Additional Definitions and their Relations}

\subsection{Multi-Message Simulation-Secure QFE}\label{sec:Apendix(Multi-Input)}

In \cref{qfe:def:simsecurity} the adversary only chooses a single message. We can adjust the experiment to allow the adversary to choose multiple messages, where each message is a quantum state of dimension $d$. In the Real world the experiment is adjusted as follows:  
\begin{align*}
    &(\rho_{\msg_1}, \dots, \rho_{\msg_n}, st) \leftarrow \Adv^{O_1(\cdot)}(\mpk)\\
    & (\rho_{\ct_i}) \leftarrow \Enc(\mpk, \rho_{\msg_i}) \quad \text{for all } i \in [n]
\end{align*}

In the Ideal world the experiment is adjusted as follows:

\begin{align*}
    &(\rho_{\msg_1}, \dots, \rho_{\msg_n}, st) \leftarrow \Adv^{O_1(\cdot)}(\mpk)\\
    & (\rho_{\ct_1}, \dots, \rho_{\ct_n}) \leftarrow \Sim(1^{\lambda}, \mpk, \mathcal{V}) \quad \text{for all } i \in [n]\\
    &\quad \text{where } \mathcal{V}= (C_f, \sk_f, C_f(\rho_{\msg_1}), \dots, C_f(\rho_{\msg_n}), 1^{d})
\end{align*}

In the classical world it is known that a non-adaptive single-message secure scheme is also secure for multiple messages. In the adaptive setting this is not the case~\cite{C:GorVaiWee12}. We show that the implication from single-message schemes to multi-message schemes in the non-adaptive setting also holds for QFE schemes. To show this we need the function secret key of the QFE scheme to be classical which is true for our scheme but might not be a requirement for every realisation of QFE. 

\begin{lemma}
\label{qfe:lem:singletomulti}
    A non-adaptive single-query simulation-secure QFE scheme with classical secret keys is also a non-adaptive single-query multi-message simulation secure QFE scheme. 
\end{lemma}

\begin{proof}
    Let $(\Setup, \Keygen, \Enc, \Dec)$ be a non-adaptive single-query simulation-secure QFE scheme with simulator $\Sim$. Then we can construct the following simulator $\Sim^*$ for the multi-message scheme:
    \begin{enumerate}
        \item Obtain $\mathcal{V}= (C_f, \sk_f, C_f(\rho_{\msg_1}), \dots, C_f(\rho_{\msg_n}), 1^{|\rho_{\msg_n}|})$ from the experiment.
        \item For every $i \in [n]$ invoke the single message simulator:
        $$ \rho_{\ct_i} \leftarrow \Sim(1^{\lambda}, \mpk, \{C_f,sk_f, C_f(\rho_{\msg_i})\})$$
        \item output $(\rho_{\ct_1}, \cdots, \rho_{\ct_n})$
    \end{enumerate}

    Let $\Adv$ be an adversary that succeeds in distinguishing the Real and Ideal world in the multi-message experiment. Then there is an adversary $\Adv^*$ that can distinguish Real and Ideal world of the single-message experiment. In the following  way a Hybrid experiment is defined for each $i \in [n]$. $\Adv^*$ receives $\mpk$ and forwards it to $\Adv$. When $\Adv$ makes a key query $C_f$ $\Adv^*$ forwards the query to it's $\Keygen$ oracle and receives $\sk_f$ which it forwards to $\Adv$. When $\Adv$ outputs $(\rho_{\msg_1}, \dots, \rho_{\msg_n})$ $\Adv^*$ encrypts messages $1 $ to $i- 1$ honestly and forwards $\rho_{\msg_i}$ to it's own experiment and receives $\rho_{\ct_i}$. Messages $i+1 $ to $ n$ are encrypted using the simulator $\Sim$. $\Adv^*$ send $(\rho_{\ct_1}, \dots, \rho_{\ct_n})$ to $\Adv$ and outputs whatever $\Adv$ outputs.   Indistinguishability between Hybrids $i$ and $i+1$ follows from the security of the single-message QFE scheme. 
\end{proof}

\begin{corollary}
    The schemes in \cref{qfe:sec:unitaryfe} and \cref{qfe:sec:polyqfe} are non-adaptive single-query multi-message simulation-secure QFE schemes.
\end{corollary}

\subsection{2-Player Security of QFE}
In this Lemma we show that a single-query secure $\QFE$ scheme is still secure if two non-communicating parties each obtain a function secret key. In this definition we consider that a single ciphertext must be split between the two non-communicating parties. A slightly different notion of security where both $B$ and $C$ obtain their own copy of the ciphertext would also be implied by a $\QFE$ scheme. 

\begin{definition}[Non-Adaptive 2-player Single-Query IND-Security for QFE]
\label{qfe:def:na2playerindsecurity}

Let $\lambda$ be the security parameter and let $\Adv = (A,B,C)$ be a QPT adversary.
   \begin{align*}
            &\Exp^{2P-IND}_{\Adv,b}(1^\lambda)\\
            &(\mpk,\msk) \leftarrow \Setup(1^\lambda)\\          
             &(\rho_{\msg_0}, \rho_{\msg_1}, \rho_{\st_A}, \rho_{\st_B}, \rho_{\st_C},C_B, C_C) \leftarrow A(\mpk)\\
            &\rho_\ct \leftarrow \Enc(\mpk, \rho_{\msg_b})\\
            & \rho_{BC} \leftarrow A(\rho_{\st_A}, \rho_\ct)\\
            & \sk_{C_B} \leftarrow \Keygen(\msk, C_B), \sk_{C_C} \leftarrow \Keygen(\msk, C_C)\\
            & b_B \leftarrow B(\mpk, \rho_\ct, \rho_{\st_B}, \sk_{C_B})\\
            & b_C \leftarrow C(\mpk, \rho_\ct, \rho_{\st_C}, \sk_{C_C})
    \end{align*}
The FE scheme is called secure if for any adversary $\Adv = (A,B,C)$ where $(\rho_{\msg_0}, \rho_{\msg_1},C_B, \rho_{\st_B})$ and $(\rho_{\msg_0}, \rho_{\msg_1},C_C, \rho_{\st_C})$ are  each admissible queries (\cref{qfe:def:admis}) it holds that
$$\Pr\left[ b_B = b_C = b \right] \leq \frac{1}{2} + negl(\lambda)$$
where the random coins are taken over the randomness of $\Adv$, $\Setup, \Keygen$ and $\Enc$.
\end{definition}

\begin{remark}
    One could obtain an adaptive security notion 
    by allowing  $B$ and $C$ to make adaptive function secret key queries themselves. 
\end{remark}

\subsubsection{Difference to Unclonable Functional Encryption Experiment.} The experiments for 2-player single-query IND-Security for $\QFE$ and the experiment for unclonable functional encryption look very similar. Note that in this experiment the function secret keys that are obtained are restricted to be admissible queries. In the unclonable functional encryption experiment the function queries are not subject to any admissibility constraint which is a much stronger notion.

\begin{lemma}
    Any non-adaptively IND-secure single-query QFE scheme (\cref{qfe:def:indsecurity}) is also a 2-player single-query IND-secure QFE scheme (\cref{qfe:def:na2playerindsecurity}).
\end{lemma}

\begin{proof}
    An adversary $\Tilde{A}$ in the single-query QFE IND-experiment can execute an adversary $(A,B,C)$ that wins the 2-player IND-experiment by only executing $A$ and $B$ and by only making a single key query $C_B$. Since to break security in the 2-player IND-security experiment both players $B$ and $C$ need to guess the correct bit $b$, $\Tilde{A}$ can win the IND-security experiment with the same probability as $(A,B,C)$ by outputting the guess $B$ outputs. 
\end{proof}

\section{Proof: Sim-security implies IND-security}
\label{qfe:app:lem1proof}

Here we provide the proof of \cref{qfe:lem:simtoind}. We restate the Lemma for convenience.
\begin{lemma}
    A QFE scheme that is single-query (non)-adaptively  SIM-secure (\cref{qfe:def:simsecurity}) is also single-query (non)-adaptively IND-secure (\cref{qfe:def:indsecurity}).
\end{lemma}

\begin{proof}
    Let $\Adv$ be an adversary that wins $\Exp^{IND}_{\Adv,b}$ with non-negligible probability. Then we can define an adversary $\Adv^*$ that wins the SIM-security experiment with non-negligible probability. Upon receiving $\mpk$ $\Adv^*$ runs $\Adv$ on input $\mpk$ until $\Adv$ outputs $(\rho^{EU}_{\msg_0}, \rho^{EU}_{\msg_1},\rho_{\st})$. A key-query of $\Adv$ is forwarded by $\Adv^*$ to it's own key oracle. Then $\Adv^*$ samples a random bit $b$ and sends $\rho^E_{\msg_b}$ as it's challenge message and receives $\rho_\ct$. $\Adv^*$ runs $\Adv$ on input $(\rho_\ct ,\rho^U_{\msg_b}, \rho_{\st})$ until it outputs a guess $b'$. 

    $\Adv^*$ outputs the state $(b',b)$. If $\Adv^*$ interacted in the ideal world the probability that $b = b'$ is $\frac{1}{2}+ negl(\lambda)$. In the ideal world the simulator receives the state $C(\rho^U_{\msg_b})$ without any information on the bit $b$. Let $\Phi$ be a completely positive trace preserving  (CPTP) map that describes the action of the simulator in the ideal experiment and $\Phi'$ be a CPTP map that applies $\Phi$ on the corresponding subsystem and the identity everywhere else.  After receiving the ciphertext the adversary holds the state $\sum_i \Phi(C(\rho^E_{\msg_{b,i}})) \otimes \rho^U_{\msg_{b,i}} \otimes \rho_{A_i}$. 

\begin{align*}
    \label{qfe:admis2}
    &\T(\sum_{i} \Phi( C(\rho^E_{\msg_{0,i}})) \otimes \rho^U_{\msg_{0,i}} \otimes \rho_{A_i}, \sum_{i} \Phi (C(\rho^E_{\msg_{1,i}})) \otimes \rho^U_{\msg_{1,i}} \otimes \rho_{A_i}) \\
    &=\T(\sum_{i} \Phi'( C(\rho^E_{\msg_{0,i}}) \otimes \rho^U_{\msg_{0,i}} \otimes \rho_{A_i}), \sum_{i} \Phi' (C(\rho^E_{\msg_{1,i}}) \otimes \rho^U_{\msg_{1,i}} \otimes \rho_{A_i}))\\
    &=\T(\Phi'(\sum_{i}  C(\rho^E_{\msg_{0,i}}) \otimes \rho^U_{\msg_{0,i}} \otimes \rho_{A_i}), \Phi' (\sum_{i} C(\rho^E_{\msg_{1,i}}) \otimes \rho^U_{\msg_{1,i}} \otimes \rho_{A_i}))\\
     &\leq \T(\sum_{i} C(\rho^E_{\msg_{0,i}}) \otimes \rho^U_{\msg_{0,i}} \otimes \rho_{A_i}, \sum_{i} C(\rho^E_{\msg_{1,i}}) \otimes \rho^U_{\msg_{1,i}} \otimes \rho_{A_i})\\
   & \leq negl(\lambda)
\end{align*}
The second to last step follows from the fact that the trace distance cannot be increased by applying a CPTP map. By definition of the trace distance $\Adv$ cannot distinguish the two states with more than negligible probability in the ideal world. 

By assumption $\Adv$ wins the IND-experiment with non-negligible advantage, therefore in the case of the real world $b=b'$ with $\frac{1}{2} + \eps$ where $\eps$ is non-negligible probability and we can distinguish the real and ideal cases with advantage $\eps/2$. 
\end{proof}

\end{document}